\documentclass[conference]{IEEEtran}
\usepackage{soul}
\usepackage[mathscr]{eucal}
\usepackage[cmex10]{amsmath}
\usepackage{epsfig,epsf,psfrag}
\usepackage{amssymb,amsmath,amsthm,amsfonts,latexsym}
\usepackage{amsmath,graphicx,bm,xcolor,url,overpic}
\usepackage{fixltx2e}
\usepackage{array}
\usepackage{verbatim}
\usepackage{bm}
\usepackage{algorithmic}
\usepackage{algorithm}
\usepackage{verbatim}
\usepackage{textcomp}
\usepackage{mathrsfs}
\usepackage{epstopdf}

\newcommand{\openone}{\leavevmode\hbox{\small1\normalsize\kern-.33em1}}

\catcode`~=11 \def\UrlSpecials{\do\~{\kern -.15em\lower .7ex\hbox{~}\kern .04em}} \catcode`~=13

\allowdisplaybreaks[4]

\newcommand{\nn}{\nonumber}

\newcommand{\calB}{\mathcal{B}}

\newcommand{\calE}{\mathcal{E}}

\newcommand{\calK}{\mathcal{K}}

\newcommand{\calP}{\mathcal{P}}

\newcommand{\calR}{\mathcal{R}}

\newcommand{\calT}{\mathcal{T}}

\newcommand{\calW}{\mathcal{W}}
\newcommand{\calX}{\mathcal{X}}
\newcommand{\calY}{\mathcal{Y}}

\newcommand{\bw}{\mathbf{w}}
\newcommand{\bW}{\mathbf{W}}
\newcommand{\bx}{\mathbf{x}}
\newcommand{\bX}{\mathbf{X}}
\newcommand{\by}{\mathbf{y}}
\newcommand{\bY}{\mathbf{Y}}

\newcommand{\rmc}{\mathrm{c}}

\newcommand{\rmU}{\mathrm{U}}

\newcommand{\bbE}{\mathsf{E}}

\newcommand{\bbN}{\mathbb{N}}

\newcommand{\bbR}{\mathbb{R}}

\DeclareMathAlphabet{\mathbsf}{OT1}{cmss}{bx}{n}
\DeclareMathAlphabet{\mathssf}{OT1}{cmss}{m}{sl}

\DeclareSymbolFont{bsfletters}{OT1}{cmss}{bx}{n}
\DeclareSymbolFont{ssfletters}{OT1}{cmss}{m}{n}
\DeclareMathSymbol{\bsfGamma}{0}{bsfletters}{'000}
\DeclareMathSymbol{\ssfGamma}{0}{ssfletters}{'000}
\DeclareMathSymbol{\bsfDelta}{0}{bsfletters}{'001}
\DeclareMathSymbol{\ssfDelta}{0}{ssfletters}{'001}
\DeclareMathSymbol{\bsfTheta}{0}{bsfletters}{'002}
\DeclareMathSymbol{\ssfTheta}{0}{ssfletters}{'002}
\DeclareMathSymbol{\bsfLambda}{0}{bsfletters}{'003}
\DeclareMathSymbol{\ssfLambda}{0}{ssfletters}{'003}
\DeclareMathSymbol{\bsfXi}{0}{bsfletters}{'004}
\DeclareMathSymbol{\ssfXi}{0}{ssfletters}{'004}
\DeclareMathSymbol{\bsfPi}{0}{bsfletters}{'005}
\DeclareMathSymbol{\ssfPi}{0}{ssfletters}{'005}
\DeclareMathSymbol{\bsfSigma}{0}{bsfletters}{'006}
\DeclareMathSymbol{\ssfSigma}{0}{ssfletters}{'006}
\DeclareMathSymbol{\bsfUpsilon}{0}{bsfletters}{'007}
\DeclareMathSymbol{\ssfUpsilon}{0}{ssfletters}{'007}
\DeclareMathSymbol{\bsfPhi}{0}{bsfletters}{'010}
\DeclareMathSymbol{\ssfPhi}{0}{ssfletters}{'010}
\DeclareMathSymbol{\bsfPsi}{0}{bsfletters}{'011}
\DeclareMathSymbol{\ssfPsi}{0}{ssfletters}{'011}
\DeclareMathSymbol{\bsfOmega}{0}{bsfletters}{'012}
\DeclareMathSymbol{\ssfOmega}{0}{ssfletters}{'012}

\newcommand{\tilK}{\tilde{K}}

\newcommand{\hatx}{\hat{x}}
\newcommand{\hatX}{\hat{X}}
\newcommand{\tilx}{\tilde{x}}
\newcommand{\tilX}{\tilde{X}}

\newcommand{\haty}{\hat{y}}
\newcommand{\hatY}{\hat{Y}}
\newcommand{\tily}{\tilde{y}}
\newcommand{\tilY}{\tilde{Y}}

\newcommand{\barx}{\bar{x}}
\newcommand{\bary}{\bar{y}}

\newtheorem{theorem}{Theorem}
\newtheorem{lemma}{Lemma}

\newtheorem{definition}{Definition}

\graphicspath{{./figures/}}
\usepackage{graphicx,cite}
\usepackage{epstopdf}
\usepackage{tikz}
\usepackage{stfloats}
\usepackage{url}
\usetikzlibrary{arrows.meta, positioning, shapes, calc, fit,backgrounds,shadows}
\usepackage{enumerate}
\usepackage{caption}
\usepackage[ colorlinks = true,
linkcolor = blue,
urlcolor  = blue,
citecolor = red,
anchorcolor = green,]{hyperref}
\usepackage{soul,color}
\usepackage{multirow}
\usepackage{amsthm}
\theoremstyle{plain}
\newtheorem{assumption}{Assumption}
\IEEEoverridecommandlockouts
\soulregister\em7
\soulregister\cite7
\soulregister\ref7
\soulregister\eqref7
\soulregister\underline7
\soulregister\emph7
\soulregister\footnote7
\definecolor{Dyellow}{RGB}{254,152,0}
\definecolor{Dgreen}{RGB}{0,176,80}

\begin{document}

\title{Rate-Distortion-Perception Tradeoff for the Gray-Wyner Problem}
\author{
Yu Yang, Yingxin Zhang, Weijie Yuan and Lin Zhou
\thanks{The authors are with School of Automation and Intelligent Manufacturing, Southern University of Science and Technology, Shenzhen, China (Emails: \{yangy2024@mail.sustech.edu.cn; zhang\_yx@buaa.edu.cn; yuanwj@sustech.edu.cn; zhoul9@sustech.edu.cn\}). Yingxin Zhang is also with the School of Cyber Science and Technology, Beihang University, Beijing, China.}
}
\maketitle

\flushbottom
\allowdisplaybreaks[1]

\begin{abstract}
We revisit the Gray-Wyner lossy source coding problem and derive the first-order asymptotic optimal rate-distortion-perception region when additional perception constraints are imposed on reproduced source sequences. The optimal trade-off is shown to be governed by a mutual information term involving common information and two conditional rate-distortion-perception functions. The perception constraint requires that the distribution of each reproduced sequence is close to that of the original source sequence, which is motivated by practical applications in image and video compression. Prior studies usually focus on the compression and reconstruction of a single source sequence. In this paper, we generalize the prior results for point-to-point systems to the representative multi-terminal setting of the Gray-Wyner problem with two correlated source sequences. In particular, we integrate the analyses of the distortion and the perception constraints by including the random circular shift operator in the encoding and decoding process directly.
\end{abstract}


\section{Introduction}\label{sec:intro}
Lossy source coding, also known as rate-distortion theory, is a fundamental problem in information theory pioneered by Shannon~\cite{Shannon_1959_Coding-theorems-for-a-discrete-source-with-a-fidelity-criterion}. The goal is to compress a source sequence and reconstruct it in a non-perfect manner, where the performance is evaluated by a distortion function that measures the difference between the original and the reconstructed source sequences, e.g., the mean squared error. However, as shown in~\cite{Blau_2018_The-Perception-Distortion-Tradeoff}, minimizing the distortion level does not necessarily lead to optimal compression performance for images in terms of human feelings such as blurriness. This issue is known as the \emph{regression to the mean} phenomenon~\cite{Blau_2018_The-Perception-Distortion-Tradeoff,Ledig_2017_Photo-Realistic-Single-Image-Super-Resolution-Using-a-Generative-Adversarial-Network}. Mathematically, the issue arises because we usually measure distortion by average statistics, which smooth out the stochastic details that are essential for human perception of images~\cite{Isola_2017_Image-to-Image-Translation-with-Conditional-Adversarial-Networks}, leading to blurriness.

To address the issue, Blau and Michaeli introduced the perception constraint and developed the rate-distortion-perception (RDP) theory~\cite{Blau_2018_The-Perception-Distortion-Tradeoff,Blau_2019_The-Rate-Distortion-Perception-Tradeoff}. The perception constraint requires that the distributions of the original and reconstructed sequences are close. In the ideal case, the reconstructed sequences should have the same distribution as the original sequence for the best perception quality. Blau and Michaeli justified the effectiveness of the perception constraint via extensive experiments and showed that the rate required to satisfy both the distortion and perception constraints is given by the rate-distortion-perception function, which resembles the rate-distortion function with an additional constraint on the distance between the distributions of the reconstructed and source symbols.

Although the theoretical benchmarks appear similar, the proof for the RDP theory requires additional efforts. The intricacies were first addressed by Theis and Wagner~\cite{Theis_2021_A-coding-theorem-for-the-rate-distortion-perception-function} and thoroughly discussed by Chen \emph{et al.}~\cite{chen_2022_on-the-rate-distortion-perception-function}. In particular, the authors of~\cite{chen_2022_on-the-rate-distortion-perception-function} identified three cases of common randomness that might lead to different theoretical benchmarks and explored several formulations of perfect perception. When common randomness is available between the encoder and the decoder, a random circular shift using the common randomness can be applied to an optimal code for rate-distortion so that the additional perception constraint is also satisfied. It was also shown that the common randomness can be removed via source simulation. A main takeaway message is that when the perception is not required to be perfect, the theoretical benchmark remains the same regardless of the availability of common randomness. In this paper, we consider non-perfect perception constraints.

The results of~\cite{Blau_2018_The-Perception-Distortion-Tradeoff,Blau_2019_The-Rate-Distortion-Perception-Tradeoff,chen_2022_on-the-rate-distortion-perception-function} have been generalized to various settings, including Gaussian sources~\cite{Qian_2023_Rate-Distortion-Perception-Tradeoff-for-Gaussian-Vector-Sources, Serra_2024_On-the-Computation-of-the-Gaussian-Rate–Distortion–Perception-Function}, the conditional perception constraint~\cite{Salehkalaibar_2024_Rate-Distortion-Perception-Tradeoff-Based-on-the-Conditional-Distribution-Perception-Measure}, the additional classification constraint~\cite{Fang_2024_The-Rate-Distortion-Perception-Classification-Tradeoff:Joint-Source-Coding-and-Modulation-via-Inverse-Domain-GANs,Wang_2025_Task-Oriented-Lossy-Compression-With-Data-Perception-and-Classification-Constraints}, the quadratic Wasserstein distance perception constraint~\cite{Qu_2025_Rate-Distortion-Perception-Theory-for-the-Quadratic-Wasserstein-Space}, the connection to lossy source coding with output constraints~\cite{Xie_2025_Output-Constrained-Lossy-Source-Coding-With-Application-to-Rate-Distortion-Perception-Theory}, computational aspects~\cite{Serra_2025_Alternating-Minimization-Schemes-for-Computing-Rate-Distortion-Perception-Functions-With-f-Divergence-Perception-Constraints}, the successive refinement setting~\cite{Zhang_2025_Universal-Rate-Distortion-Perception-Representations-for-Lossy-Compression}, the Wyner-Ziv setting~\cite{Hamdi_2023_The-Rate-Distortion-Perception-Trade-off-with-Side-Information} and the joint source-channel coding setting~\cite{Tan_2025_Rate-Distortion-Perception-Controllable-Joint-Source-Channel-Coding-for-High-Fidelity-Generative-Semantic-Communications}.

Although the above results provide insights, all these results are restricted to the point-to-point setting that compresses a single sequence. However, in many practical applications, e.g., stereo image compression~\cite{Deng_2021_Deep-Homography-for-Efficient-Stereo-Image-Compression} and neural distributed source coding~\cite{Mital_2022_Neural-Distributed-Image-Compression-Using-Common-Information}, multiple source sequences need to be compressed at one place and reconstructed at different locations. A typical model for this scenario is the Gray-Wyner problem~\cite{Gray_1974_Source-coding-for-a-simple-network}, where two correlated source sequences are compressed at one location and are required to be recovered at two different locations. Specifically, the Gray-Wyner model compresses two correlated source sequences into a common message and two private messages. The common message captures the correlated part of two source sequences while the two private messages convey the additional unique part of two source sequences. Such a layered coding framework helps reduce the compression rates by using the correlation between source sequences.

The rate-distortion region for the Gray-Wyner problem has been well studied. Specifically, for discrete source sequences, the first-order asymptotic region was derived by Gray and Wyner~\cite{Gray_1974_Source-coding-for-a-simple-network} in their pioneering paper and recently refined by Watanabe~\cite{Watanabe_2017_Second-Order-Region-for-Gray–Wyner-Network} for the lossless case and by Zhou, Tan and Motani~\cite{Zhou_2017_Discrete-Lossy-Gray–Wyner-Revisited} for the lossy case. The common information, defined as the minimal common rate subject to a constraint on the sum rate, was studied by Xu, Liu and Chen~\cite{Xu_2016_A-Lossy-Source-Coding-Interpretation-of-Wyner's-Common-Information} and Viswanatha, Akyol and Rose~\cite{Viswanatha_2014_The-Lossy-Common-Information-of-Correlated-Sources}. Explicit formulas for the rate-distortion region were recently provided by Yu~\cite{Yu_2023_Gray–Wyner-and-Mutual-Information-Regions-for-Doubly-Symmetric-Binary-Sources-and-Gaussian-Sources} for both binary and Gaussian correlated source sequences.

However, the perception constraint has not been studied for the Gray-Wyner problem yet. In this paper, we fill the above research gap and derive the first-order asymptotic optimal rate-distortion-perception region for the Gray-Wyner problem with both distortion and perception constraints. Specifically, we show that the optimal trade-off is governed by the conditional RDP functions for private messages and a mutual information term for the common message. Furthermore, our achievability proof generalizes the ideas in~\cite{chen_2022_on-the-rate-distortion-perception-function} and streamlines the proof procedures by integrating the random shift operator directly into the encoding and decoding process with common randomness. Instead of using a rate-distortion code and adding the perception analysis subsequently, we analyze both distortion and perception constraints for a coding scheme tailored for the RDP Gray-Wyner problem. We would like to emphasize that the generalization from the point-to-point setting~\cite{chen_2022_on-the-rate-distortion-perception-function} to the Gray-Wyner problem is non-trivial. This is mainly because the rate-distortion region of the Gray-Wyner problem includes an auxiliary random variable that helps characterize the rate of the common message.

\begin{figure}[t]
\centering
\begin{tikzpicture}[>=latex, thick, font=\footnotesize]
\def\blockW{1.2cm}  
\def\blockH{0.7cm}  
\def\vSep{1.2}     
\def\hSep{1.7cm}    
\def\outerW{1cm}   
\def\outerH{3.6cm}   
\node[draw, rounded corners, minimum width=1.8cm, minimum height=\outerH] (encoder) at (0,0) {};
\node[draw, rounded corners, minimum width=1.8cm, minimum height=\outerH, right=1cm of encoder] (decoder) {}; 
\node[draw, rounded corners, minimum width=\blockW, minimum height=\blockH] (f1) at (0, \vSep) {$f_{1}$};
\node[draw, rounded corners, minimum width=\blockW, minimum height=\blockH] (f0) at (0, 0.0)   {$f_{0}$};
\node[draw, rounded corners, minimum width=\blockW, minimum height=\blockH] (f2) at (0,-\vSep) {$f_{2}$};
\node[draw, rounded corners, minimum width=\blockW, minimum height=\blockH, right=\hSep of f1] (phi1) {$\phi_{1}$};
\node[draw, rounded corners, minimum width=\blockW, minimum height=\blockH, right=\hSep of f2] (phi2) {$\phi_{2}$};
\coordinate (J)  at ($(f0.west)+(-0.5,0)$);
\coordinate (Ju) at ($(J)+(0, \vSep)$);
\coordinate (Jd) at ($(J)+(0,-\vSep)$);

\node[left] at ($(J)+(0.1,0.2)$) {$(X^n,Y^n)$}; 
\draw ($(J)+(-1.0,0)$) -- (J); 
\draw (J) -- (Ju);
\draw (J) -- (Jd);

\draw[-{Latex}] (Ju) -- (f1.west);
\draw[-{Latex}] (J)  -- (f0.west);
\draw[-{Latex}] (Jd) -- (f2.west);

\coordinate (M1) at ($(f1.east)!0.5!(phi1.west)$);
\coordinate (M2) at ($(f2.east)!0.5!(phi2.west)$);
\coordinate (M0) at ($(f0.east)!0.5!(f0.east-|phi1.west)$); 

\draw[-{Latex}] (f1.east) -- (phi1.west);
\node[above, inner sep=1pt, xshift=-0.1cm] at (M1) {$S_{1}$};

\draw[-{Latex}] (f2.east) -- (phi2.west);
\node[above, inner sep=1pt, xshift=-0.1cm] at (M2) {$S_{2}$};

\draw[-{Latex}] (phi1.south) -- (phi2.north);
\draw[-{Latex}] (phi2.north) -- (phi1.south);

\coordinate (Mid) at ($(phi1.south)!0.5!(phi2.north)$);
\draw (f0.east) -- (Mid-|M0) -- (Mid); 
\node[above, inner sep=1pt, xshift=-0.1cm] at (Mid-|M0) {$S_{0}$};

\draw[-{Latex}] (phi1.east) -- ++(2.1,0) node[above,xshift=-1cm] {$(\hat{X}^{n}, D_{1}, P_{1})$};
\draw[-{Latex}] (phi2.east) -- ++(2.1,0) node[above,xshift=-1cm] {$(\hat{Y}^{n}, D_{2}, P_{2})$};

\draw[->] ($(encoder.north)+(0,0.4)$) node[above, inner sep=1pt] {$K$} -- (encoder.north);
\draw[->] ($(decoder.north)+(0,0.4)$) node[above, inner sep=1pt] {$\widetilde{K}$} -- (decoder.north);
\end{tikzpicture}
\caption{Lossy Gray-Wyner model with perception and randomness.}
\label{fig:encoder-decoder}
\end{figure}
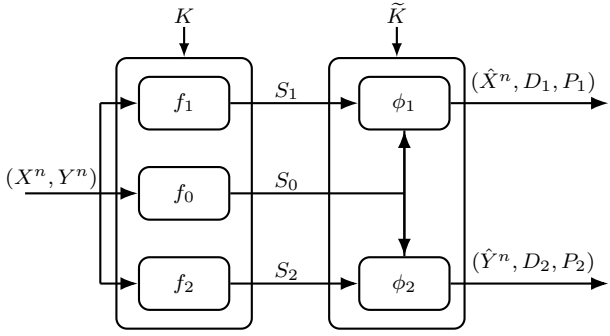

\section{Problem Formulation and Definitions}
\subsection*{Notation}
Random variables are in capital case (e.g., $X$) and their realizations are in lower case (e.g., $x$). We use calligraphic font (e.g., $\mathcal{X}$) to denote all sets. Random vectors of length $n$ and their particular realizations are denoted by $X^n:= (X_1, \ldots, X_n)$ and $x^n=(x_1,\ldots,x_n)$, respectively. Given any sets $\calX$ and $\calY$, we use $\calX\setminus\calY$ to denote the set difference, $|\calX|$ to denote the cardinality of $\calX$ and use $\calX^{\rmc}$ to denote the complement of $\calX$. We use $\bbR_+$, $\bbN$, $\bbN_+$ to denote the sets of positive real numbers, integers and positive integers respectively. For any two integers $(a,b)\in\bbN^2$, we use $[a:b]$ to denote the set of integers between $a$ and $b$, and we use $[a]$ to denote $[1:a]$. We use $\log$ and $\exp$ with base $2$. We use $\mathrm{Var}$ to denote variance. We use $\mathrm{U}$ to denote the uniform distribution.

The set of all probability distributions on a set $\calX$ is denoted as $\calP(\calX)$. Given any two sequences $(x^n,y^n)\in\calX^n\times\calY^n$, the empirical distribution of $x^n$ is denoted as $\hat{T}_{x^n}$, the joint empirical distribution of $(x^n, y^n)$ is denoted as $\hat{T}_{x^n, y^n}$, and the conditional empirical distribution of $x^n$ given $y^n$ is denoted as $\hat{T}_{x^n|y^n}$.
Given any two distributions $(P,Q)\in \calP(\calX)^2$, we use $d_{\rm{TV}}(P,Q)=\sum_{x}|P(x)-Q(x)|$ to denote the TV distance. We use $\mathbf{1}\{\}$ to denote the indicator function. Finally, we follow \cite{Gamal_2011_Network-Information-Theory} for notation of information-theoretic quantities.

\subsection{Problem Formulation}
Fix a joint distribution $P_{XY}$ defined on the finite alphabet $\calX\times\calY$. Consider memoryless source sequences $(X^n,Y^n)$ that are generated i.i.d. from $P_{XY}$. 
Let $(K, \tilde{K})$ be random variables taking values in alphabets $\calK$ and $\tilde{\calK}$, respectively.  Fix three integers $(M_0,M_1,M_2) \in \bbN^3$ and four positive real numbers $(P_1, P_2,D_1,D_2)\in\mathbb{R}_+^4$. As shown in Fig.~\ref{fig:encoder-decoder}, in the Gray-Wyner problem with both distortion and perception constraints, one aims to compress the correlated source sequences $(X^n,Y^n)$ into a common message $S_0\in[M_0]$ and two private messages $(S_1,S_2)\in[M_1]\times[M_2]$ such that the source sequences $X^n$ and $Y^n$ are reconstructed as $(\hat{X}^n,\hat{Y}^n)$ within distortion-perception levels $(D_1,P_1)$ and $(D_2,P_2)$, respectively. The encoding is done via three encoders $(f_0,f_1,f_2)$ and decoding is performed by two decoders $(\phi_1,\phi_2)$. Note that a pair of random variables $(K,\tilK)$ taking values in the alphabets $(\calK,\tilde{\calK})$ is available at the encoder and the decoder, respectively. When $K=\tilK$, the setting is known as common randomness; when $K\neq\tilK$, the setting is known as private randomness; when $(K,\tilK)$ are both constants, the setting is known as no randomness. We consider common randomness in this paper, where the common random seed $K=\tilK$ is assumed to be uniformly distributed over $[0:n-1]$. Consistent with~\cite{chen_2022_on-the-rate-distortion-perception-function}, under the non-perfect perception constraint, the common randomness assumption can be eliminated via source simulation.
 
In short, one aims to encode a pair of correlated source sequences into two private messages and one common message such that two independent decoders, each having access to the common message and only one of the private messages, can reliably reconstruct one component of the original source pair under given distortion constraints and perception constraints. 


\subsection{Definitions}
Let $\{\mathrm{cr,pr,nr}\}$ denote common randomness, private randomness, and no randomness, respectively. Fix any $\dagger\in\{\mathrm{cr,pr,nr}\}$ and positive integer $n\in\bbN_+$. A code is formally defined as follows.
\begin{definition}
\label{def:code}
An $(n,M_0,M_1,M_2,\dagger)$-code for the Gray-Wyner problem with both distortion and perception constraints consists of three encoders:
\begin{align}
f_{i} &: \ \calX^n \times \calY^n \times\calK \to [M_i],\, i\in[0:2] \label{eq:encoder}
\end{align}
and two decoders:
\begin{align}
\phi_{1} &: \ [M_0]\times[M_1] \times\tilde\calK\to \hat{\calX}^n, \label{eq:phi1}\\
\phi_{2} &: \ [M_0]\times[M_2] \times\tilde\calK\to \hat{\calY}^n. \label{eq:phi2}
\end{align}
\end{definition}
To evaluate the distortion constraints, consider the following two bounded distortion measures: $\Delta_1:\mathcal{X}\times\hat{\mathcal{X}}\to[0,\infty)$,  $\Delta_2:\mathcal{Y}\times\hat{\mathcal{Y}}\to[0,\infty)$
such that for each $(x,y)\in\mathcal{X}\times\mathcal{Y}$, there exists $(\hat x,\hat y)\in\hat{\mathcal{X}}\times\hat{\mathcal{Y}}$ satisfying
$\Delta_1(x,\hat x)=0$ and $\Delta_2(y,\hat y)=0$. To evaluate the perception constraints, consider the following two perception measures: $d_1:\calP(\calX)\times\calP(\calX)\to[0,\infty)$,  $d_2:\calP(\calY)\times\calP(\calY)\to[0,\infty)$ such that $d_1(P_X,P_{\hat{X}})=0$ if and only if (iff) $P_X=P_{\hat{X}}$ and $d_2(P_Y,P_{\hat{Y}})=0$ iff $P_Y=P_{\hat{Y}}$.

Note that the joint distribution of the correlated source sequences $(X^n,Y^n)$, the messages $(S_0,S_1,S_2)$, and the reconstructed sequences $(\hatX^n,\hatY^n)$ is induced by the source distribution $P_{XY}$ and an $(n,M_0,M_1,M_2,\dagger)$-code per Def.~\ref{def:code}. Any other distribution, including the marginal distribution of $(X_t,Y_t,\hatX_t,\hatY_t)$, is induced by this joint distribution. The rate-distortion-perception region is defined as follows.
\begin{definition}\label{definition:2}
Given any tuple of non-negative real numbers $(R_0,R_1,R_2)\in\bbR_+^3$, the rate triple $(R_0,R_1,R_2)$ is said to be $(D_1,D_2,P_1,P_2)$-achievable if there exists a sequence of $(n,M_0,M_1,M_2,\dagger)$-codes such that
\begin{itemize}
\item for all $i\in\{0,1,2\}$,
\begin{align}
\limsup_{n\to\infty}\frac{1}{n}\log M_i&\leq R_i,
\end{align}
\item the following distortion constraints are satisfied
\begin{align}
\frac{1}{n}\sum_{t\in[n]}\bbE_{P_{X_t\hat{X}_t}}\Big[\Delta_1(X_t,\hat X_t)\Big]&\leq D_1,\\
\frac{1}{n}\sum_{t\in[n]}\bbE_{P_{Y_t\hat Y_t}}\Big[\Delta_2(Y_t,\hat Y_t)\Big]&\leq D_2,
\end{align}
\item and the following perception constraints are satisfied for each $t\in [n]$,
\begin{align}
d_1(P_{X_t},P_{{\hatX}_t})&\leq P_1, \\
d_2(P_{Y_t},P_{{\hatY}_t})&\leq P_2.
\end{align}
\end{itemize}
The closure of the set of all $(D_1,D_2,P_1,P_2)$-achievable rate triplets is the $(D_1,D_2,P_1,P_2)$-achievable rate region and denoted as $\calR_{\dagger}(D_1,D_2,P_1,P_2)$.
\end{definition}
\section{Main Results}
\subsection{Preliminaries}
Recall that $(\calX,\calY)$ are the source alphabets. Let $\calW$ be another finite set and define $\calP(P_{XY})$ as the set of all joint distributions $Q_{XYW}\in\calP(\calX\times\calY\times\calW)$ such that $Q_{XY}=P_{XY}$ and $|\calW| \leq |\calX||\calY|+2$. Fix any distribution $Q_{XYW}$. Let $(Q_{XW},Q_{YW})$ be the induced distributions. Given any non-negative real numbers $(D_1,P_1)\in\mathbb R_+^2$, define the conditional
RDP function of $X$ given $W$ as follows:
\begin{align}
&R_{X|W}(Q_{XW}, D_1, P_1) \nonumber\\&:= \min_{\substack{Q_{\hat{X}|XW}: \, \bbE_{Q_{X\hat{X}}}[\Delta_1(X, \hat{X})] \le D_1, \\ d_1(P_X, Q_{\hat{X}}) \le P_1}} I(X; \hat{X}|W)\label{conditional RDP},
\end{align}
where the minimization is over $Q_{\hat{X}|XW}$ such that the induced joint distribution $Q_{XW\hat{X}}=Q_{XW}Q_{\hat{X}|XW}$ satisfies the distortion and perception constraints in \eqref{conditional RDP}. Similarly, define $R_{Y|W}(Q_{YW},D_2,P_2)$. Note that when $W$ is a constant, the conditional RDP function reduces to the RDP function for the point-to-point case~\cite[Eq. (4)]{Blau_2019_The-Rate-Distortion-Perception-Tradeoff},~\cite[Eq. (1), (2)]{chen_2022_on-the-rate-distortion-perception-function}.

Analogous to \cite{chen_2022_on-the-rate-distortion-perception-function}, we need  the following assumptions. The first assumption ensures that the RDP functions are bounded so that one can construct a coding scheme to satisfy both the distortion and perception constraints.
\begin{assumption}\label{assump:1}
Given any $(D_1,D_2,P_1,P_2)\in\bbR_+^4$, for any $Q_{XYW}\in\calP(P_{XY})$, assume that i) both $R_{X|W}(Q_{XW},D_1, P_1)$ and $R_{Y|W}(Q_{YW},D_2,P_2)$ are finite and ii) given any non-negative real number $\varepsilon\in\bbR_+$, there exist discrete random variables $(\hatX,\hatY)\in\hat{\calX}\times\hat{\calY}$ with $\hat{\mathcal{X}} \subseteq \mathcal{X}$ and $\hat{\mathcal{Y}} \subseteq \calY$ and distributions $(Q_{\hatX|XW},Q_{\hatY|YW})$ such that
\begin{align}
I(X; \hatX|W) &\le R_{X|W}(Q_{XW},D_1, P_1)+\varepsilon,\label{assump_1_1}\\
I(Y; \hatY|W) &\le R_{Y|W}(Q_{YW},D_2, P_2)+\varepsilon\label{assump_1_2},\\
\bbE_{Q_{X\hat X}}[\Delta_1(X,\hat{X})] &\le D_1 + \varepsilon,\label{assump_1_3}\\
\bbE_{Q_{Y\hat Y}}[\Delta_2(Y,\hat{Y})] &\le D_2 + \varepsilon,\label{assump_1_4}\\
d_1(P_X,Q_{\hat{X}}) &\le P_1 + \varepsilon,\label{assump_1_5}\\
d_2(P_Y,Q_{\hat{Y}}) &\le P_2 + \varepsilon,\label{assump_1_6}
\end{align}
where the joint distribution of random variables $(X,Y,W,\hatX,\hatY)$ is given by $Q_{XYW}\times Q_{\hatX|XW}\times Q_{\hatY|YW}$.
\end{assumption}

We next clarify the validity of these assumptions. Bounded assumptions for conditional RDP functions are prerequisites for the existence of a coding theorem to achieve distortion and perception constraints $(D_1,D_2,P_1,P_2)$. These two assumptions are naturally satisfied since for finite alphabets $(\calX,\calY)$, $R_{X|W}(Q_{XW},D_1, P_1) \le R_{X|W}(Q_{XW},0, 0) = H(X|W) \leq H(X)< \infty$ and similarly $R_{Y|W}(Q_{YW},D_2, P_2)\leq H(Y)< \infty$. Conditions in \eqref{assump_1_1} to \eqref{assump_1_6} also hold for discrete source sequences. Since $R_{X|W}(Q_{XW},D_1, P_1)$ is bounded, the optimization problem in \eqref{conditional RDP} is feasible, which implies the existence of a random variable $\tilde{X}$ such that 
$I(X; \tilde{X}|W) \le R_{X|W}(Q_{XW},D_1,P_1) + \varepsilon$, $\bbE_{Q_{X\tilde X}}[\Delta_1(X,\tilde{X})] \le D_1$, and $d_1(P_X, Q_{\tilde{X}}) \le P_1$. Similar logic holds for the $Y$ part. In summary, Assumption~\ref{assump:1} is typically satisfied in most practical scenarios, and always holds when $|\mathcal{X}| < \infty$. The assumptions \(\hat{\mathcal{X}} \subseteq \mathcal{X}\) and \(\hat{\mathcal{Y}} \subseteq \mathcal{Y}\) arise from the perception constraint. This is because to satisfy the perception constraint that requires negligible difference between the source distribution and the reproduced distribution, the reconstructed distributions should have the same support set as the source distributions.

The second assumption concerns the convexity of the perception function.
\begin{assumption}\label{assump:2}
For each $i\in[2]$, assume that the perception function $d_i(\cdot,\cdot)$ is convex in its second argument.
\end{assumption}
This assumption is mild as it is satisfied by the general class of \(f\)-divergence~\cite[Section 4]{Csiszar_2004_Information-theory-and-statistics}, which is defined by a convex generator function \(f\) and includes the TV distance as a special case (see Example (5) in~\cite[Section 4]{Csiszar_2004_Information-theory-and-statistics}).

\subsection{Results and Discussions}

\begin{theorem}\label{theorem:1}
Fix any $(D_1,D_2,P_1,P_2)\in\bbR_+^4$ and $\dagger\in\{\mathrm{cr,pr,nr}\}$. Under Assumptions \ref{assump:1} and \ref{assump:2}, the RDP region for the Gray-Wyner problem satisfies 
\begin{align}\label{theorem1_result}
\begin{split}
&\mathcal{R}_{\dagger}(D_1, D_2, P_1, P_2) \\
&= \bigcup_{Q_{XYW} \in \mathcal{P}(P_{XY})} \Big\{  (R_0,R_1,R_2): R_0 \geq I(X,Y;W), \\
&\quad R_1 \geq R_{X|W}(Q_{XW},D_1,P_1), R_2 \geq R_{Y|W}(Q_{YW},D_2,P_2)\Big\}.
\end{split}
\end{align}
\end{theorem}

The proof of Theorem \ref{theorem:1} is available in our journal version \cite{Yang_2026_Rate-Distortion-Perception-Tradeoff-for-the-Gray-Wyner-Problem}. In Section~\ref{Achievability Proof}, we provide a sketch of the achievability proof. We summarize the high-level ideas for the proof steps. For the achievability part, we use a layered coding scheme where private descriptions refine the common part that extracts the correlation of the source sequences. All codewords are generated from uniform distributions over typical sets or conditional typical sets. To satisfy the perception constraint, we use the random circular shift operator that helps ensure the reconstructed marginal distributions match the target distributions. Consistent with \cite[Theorem 2]{chen_2022_on-the-rate-distortion-perception-function}, we first assume common randomness and subsequently remove the common randomness via source simulation. The converse part is based on the standard converse proof for the Gray-Wyner problem, with modifications required to include the perception constraints.

We make the following remarks. Firstly, the RDP region remains the same regardless of the availability of randomness. This is done by proving the converse part for common randomness and the achievability for no randomness. Thus, deterministic encoders and decoders are optimal in the sense of achieving the first-order asymptotic rate region. Theorem \ref{theorem:1} confirms that the perception constraint maintains the fundamental separation structure of the Gray-Wyner network. Specifically, as shown in \eqref{theorem1_result}, the Gray-Wyner problem effectively captures source correlation via the common message, while satisfying the distortion and perception constraints through two independent conditional coding steps. As a sanity check, when the perception constraints are relaxed such that \( P_1 = P_2 = \infty \), the RDP region for the Gray-Wyner problem reduces to the rate-distortion region of the lossy Gray-Wyner problem~\cite[Theorem 8]{Gray_1974_Source-coding-for-a-simple-network}.

Secondly, we clarify the relationship of our results with the recent study by Zhang \emph{et al.}~\cite{Zhang_2025_Universal-Rate-Distortion-Perception-Representations-for-Lossy-Compression} for RDP of the successive refinement problem~\cite{EquitzC_1991_Successive-refinement-of-information}. The successive refinement problem can be viewed as a degraded special case of the Gray-Wyner model. Specifically, when $Y^n$ is constant and the private encoder $f_2$ is removed from the Gray-Wyner problem in Fig.~\ref{fig:encoder-decoder}, the Gray-Wyner problem reduces to the successive refinement problem. In this setup, decoder $\phi_2$ receives only the common message $S_0$ to generate a coarse reconstruction, while decoder $\phi_1$ receives both $S_0$ and the private description $S_1$ to generate a refined reconstruction. Consequently, our common rate corresponds to the base rate, and the private rate captures the incremental information required for refinement.

\section{Achievability Proof of Theorem~\ref{theorem:1}}\label{Achievability Proof}
We provide a proof sketch for the achievability part of Theorem \ref{theorem:1}, which is decomposed into codebook generation, coding scheme, distortion analysis, perception analysis and common randomness removal. At a high level, we construct a coding scheme using common randomness where all codewords are generated uniformly from (conditional) typical sets. Using the random coding idea, we show that the constructed coding scheme satisfies both the distortion and perception constraints. Finally, following the idea of de-randomization in \cite{chen_2022_on-the-rate-distortion-perception-function}, we show that the same RDP region can be achieved when common randomness is eliminated with negligible rate overhead.

To satisfy the perception constraint, consistent with~\cite{chen_2022_on-the-rate-distortion-perception-function}, we need the following random shift operator.
\begin{definition}
Given any two integers $k \in [0:n-1]$ and $t\in[n]$, we define
\begin{align}
\theta_k^{(n)}(t) := \left((t+k-1) \bmod n\right) + 1.\label{definition:3}
\end{align}
For any pair of source sequences $(x^n,y^n)\in\calX^n\times\calY^n$, let $\pi_k(x^n,y^n)$ be a circular shift operator such that 
\begin{align}
\begin{split}
\pi_k(x^n,y^n) := \Big( & \big(x_{\theta_k^{(n)}(1)},y_{\theta_k^{(n)}(1)}\big), \ldots, \\
& \big(x_{\theta_k^{(n)}(n)},y_{\theta_k^{(n)}(n)}\big)\Big).
\end{split}
\end{align}
\end{definition}
Note that $t\in[n]$ indexes the coordinate of a length-$n$ sequence, while $k\in[0:n-1]$ is the circular shift seed. The circular shift $\pi_k$ cyclically reorders the positions using $\theta_k^{(n)}(t)$, which introduces randomness to equalize the marginal distribution of each symbol of the reproduced sequence.
\\
\subsubsection{Codebook Generation}
Fix any $Q_{XYW}\in\calP(P_{XY})$. Let $Q_{\tilX|XW}$ and $Q_{\tilY|YW}$ be the optimizers for the conditional RDP function (cf. \eqref{conditional RDP}). Define the joint distribution $Q_{XYW\tilX\tilY}:=P_{XY}Q_{W|XY}Q_{\tilde X|XW}Q_{\tilde Y|YW}$.  Fix an integer $n\in\bbN$. Let $(M_0,M_1,M_2)\in\bbN^3$ be three integers to be specified later. Firstly, we generate $M_0$ codewords $\bw:=(w^n(1),\ldots,w^n(M_0))$ independently from a uniform distribution over the typical set $\calT_{\delta}^{(n)}(Q_W)$~\cite[Section 2.4]{Gamal_2011_Network-Information-Theory}. Subsequently, for each $i\in[M_0]$, we generate $M_1$ codewords $\{\tilx^n(i,1),\ldots,\tilx^n(i,M_1)\}$ independently from the uniform distribution over the conditional typical set $\calT_{\delta}^{(n)}(Q_{\tilX|W}|w^n(i))$~\cite[Section 2.5]{Gamal_2011_Network-Information-Theory} and generate $M_2$ codewords $\{\tily^n(i,1),\ldots,\tily^n(i,M_2)\}$ independently from the uniform distribution over the conditional typical set $\calT_{\delta}^{(n)}(Q_{\tilY|W}|w^n(i))$~\cite[Section 2.5]{Gamal_2011_Network-Information-Theory}.\\
\subsubsection{Coding Scheme}
Fix any $\delta\in\bbR_+$. Assume that the source sequence realization is $(x^n, y^n)\in\calX^n\times\calY^n$ and the common randomness realization is $k\in[0:n-1]$. Before compression, the encoders first compute the circularly shifted versions of the source sequences as $(\bar{x}^n, \bar{y}^n) = \pi_{-k}(x^n, y^n)$. Subsequently, the encoder $f_0$ selects an index $s_0 \in [M_0]$ such that $(\bar{x}^n, \bar{y}^n, w^n(s_0))$ are jointly typical. If none exists, $s_0$ is set to $1$; if there is more than one such index, $s_0$ is chosen as the smallest among all. Given $s_0$,  encoders $f_1$ and $f_2$ select indices $s_1 \in [M_1]$ and $s_2 \in [M_2]$ such that 
\begin{align} 
\frac{1}{n}\sum_{t\in[n]}\Delta_1(\barx_t,\tilx_t(s_0,s_1)) \leq \bbE[\Delta_1(X,\tilX)] + \frac{\delta}{2},\label{Encoder:f1}\\
\frac{1}{n}\sum_{t\in[n]}\Delta_2(\bary_t,\tily_t(s_0,s_2)) \leq \bbE[\Delta_2(Y,\tilY)] + \frac{\delta}{2}.\label{Encoder:f2}
\end{align}
Again, if none exists, the index is set to $1$, and if there is more than one index, it is chosen as the smallest among all.

Upon receiving $(s_0,s_1)$, the decoder $\phi_1$ finds the codeword $\tilx^n(s_0,s_1)$ and outputs the reproduced version as $\hatx^n=\pi_k(\tilx^n(s_0,s_1))$ using the circular shift operator. Similarly, upon receiving $(s_0,s_2)$, the decoder $\phi_2$ finds the codeword $\tily^n(s_0,s_2)$ and outputs the reproduced version as $\haty^n=\pi_k(\tily^n(s_0,s_2))$ using the circular shift operator.\\
\subsubsection{Distortion Analysis}
Define the error event for the encoder $f_0$ as
\begin{align}
\calE_0&:= \{(\bar{X}^n,\bar{Y}^n,W^n(s_0))\notin {\calT_{\delta}}^{(n)}(Q_{XYW}),\,\forall~{s_0}\in [M_0] \}.
\end{align}
Fix any positive real number $(\delta_0,\varepsilon_0)\in\bbR_+^2$ and set $M_0 := \lfloor\exp\left(n(I(X,Y;W) + 2\delta_0)\right)\rfloor$. By the standard typicality argument \cite[Lemma 3.3]{Gamal_2011_Network-Information-Theory}, for sufficiently large $n$,
\begin{align}
\Pr\{\calE_0\} \leq \varepsilon_0.
\end{align}
Thus, asymptotically with high probability, the encoder $f_1$ can proceed successfully.

In what follows, we analyze the distortion constraint for the source sequence $X^n$ since the analysis for $Y^n$ is analogous. Fix any positive real number $\delta_1\in\bbR_+$. Set $M_1 := \lfloor \exp(n(I(X;\tilX|W) + \delta_1)) \rfloor$. Following similar analysis to the rate-distortion region proof in \cite[Theorem 8]{Gray_1974_Source-coding-for-a-simple-network}, we can prove that with the chosen $M_1$, asymptotically, the distortion constraint is satisfied such that for sufficiently large $n$,
\begin{align}
\frac{1}{n} \sum_{t\in[n]} \bbE[\Delta_1(\bar{X}_t, \tilX_t)] \leq D_1.
\end{align}
Finally, since common randomness is assumed to be uniformly distributed, which leads to the uniformity of the random circular shift, the expected distortion for the original source sequence is invariant:
\begin{align}
\frac{1}{n}\sum_{t\in[n]} \bbE[\Delta_1(X_t,\hatX_{t})] &= \frac{1}{n}\sum_{t\in[n]} \bbE[\Delta_1(\bar{X}_t,\tilX_{t})] \\
&\le D_1.
\end{align}
\subsubsection{Perception Analysis} 
Again, we analyze the perception constraint for the source sequence $X^n$ since the analysis for $Y^n$ is analogous. Since $K$ is uniformly distributed over $[0:n-1]$, the marginal distribution $Q_{\hatX_t}$ is identical for all $t \in [n]$ and equals the expected empirical distribution of $\tilX^n(S_0,S_1)$. Since $\tilX^n(S_0,S_1)$ is generated uniformly over the conditional typical set, its empirical distribution converges to $Q_{\tilX}$ in total variation distance as $n \to \infty$. Consequently, $d_{\mathrm{TV}}(Q_{\hat{X}_t}, Q_{\tilX}) \to 0$ uniformly for all $t$. Furthermore, invoking Assumption 1 which implies that $Q_{\tilX}$ satisfies $d_1(P_X, Q_{\tilX}) \le P_1 + \varepsilon$ for arbitrarily small $\varepsilon$, and Assumption 2 which ensures the continuity of the perception function $d_1(P_X,Q)$ in its second argument, it follows that asymptotically,
\begin{align}
d_1(P_X, Q_{\hat{X}_t}) \leq P_1.
\end{align}
\subsubsection{Common Randomness Removal}
To remove the common randomness, we can use source simulation to generate a uniform random variable using a negligible fraction of the source sequence. Specifically, for any positive real number $\alpha\in\bbR_+$, set $n_0 := \lfloor n\alpha \rfloor$; it follows from \cite[Appendix B]{chen_2022_on-the-rate-distortion-perception-function} that we can find a mapping $\omega: \mathcal{X}^{n_0}\times\mathcal{Y}^{n_0} \to [0 : n-1]$ such that the distribution of $\tilde{K} = \omega(X_{n+1}^{n+n_0}, Y_{n+1}^{n+n_0})$ is exponentially close to the uniform distribution over $[0:n-1]$ as $n \to \infty$.

Using $\tilde{K}$ instead of the true random seed $K$, we construct a deterministic code of blocklength $n+n_0$. Let $\check{X}^{n+n_0}$ denote the reproduced source sequence generated by the aforementioned construction. It follows that
\begin{align}
\frac{1}{n+n_0} \sum_{t\in[n+n_0]}\bbE[\Delta_1(X_t, \check{X}_t)] \leq D_1,
\end{align}
for sufficiently large $n$ and small $\alpha$. The feasibility of choosing such a small $\alpha$ relies on the key observation that $K$ can be simulated using a negligible fraction of source symbols as $H(K)$ is sublinear in $n$. Furthermore, choosing $Q_{\check{X}_t|\tilde{K}} = Q_{\hat{X}_t|K}$ for $t \in [n]$, asymptotically as $n \to \infty$, the source simulation result implies that $d_{\mathrm{TV}}(Q_{\tilde{K}}, Q_K) \to 0$. Together with the triangle inequality, we have $d_{\mathrm{TV}}(Q_{\check{X}_t}, Q_{\hat{X}_t}) \to 0$ uniformly for all $t \in [n]$. Using the convergence of $Q_{\hat{X}_t}$ to $Q_{\tilX}$ in total variation distance as $\delta \to 0$ and the continuity of $d_1(P_X, Q)$ in its second argument at $Q = Q_{\tilX}$, we have
\begin{align}
d_1(P_X, Q_{\check{X}_t}) \leq P_1, \quad t \in [n+n_0].
\end{align}
Finally, invoking the fact that $\frac{\log n}{n+n_0} \to 0$ as $n \to \infty$, we conclude that the deterministic scheme achieves the rate region stated in Theorem \ref{theorem:1}. This completes the proof.

\section{Conclusion}
We revisited the Gray-Wyner problem and derived the rate-distortion-perception region. Our result generalized the RDP theory from point-to-point settings to multi-terminal distributed networks, demonstrating that deterministic coding schemes are sufficient to achieve the optimal trade-off among rate, fidelity, and perceptual quality. In particular, our proof streamlined the analyses of both distortion and perception
constraints, addressing the separation analyses in \cite{chen_2022_on-the-rate-distortion-perception-function}. In the future, one can generalize our results to the Gaussian case \cite{Qian_2023_Rate-Distortion-Perception-Tradeoff-for-Gaussian-Vector-Sources}, study the finite blocklength performance \cite{Zhou_2023_Finite-blocklength-lossy-source-coding-for-discrete-memoryless-sources} and further consider
the classification task \cite{Fang_2024_The-Rate-Distortion-Perception-Classification-Tradeoff:Joint-Source-Coding-and-Modulation-via-Inverse-Domain-GANs, Wang_2025_Task-Oriented-Lossy-Compression-With-Data-Perception-and-Classification-Constraints}.
\bibliographystyle{IEEEtran}
\bibliography{isit_v3.bib}
\cleardoublepage
\onecolumn

\appendices
\section{Proof of First-Order Rate Region (Theorem~\ref{theorem:1})} 
Now we proceed to prove Theorem \ref{theorem:1}. Let the joint distribution corresponding to a correlated source $(X, Y)$ be $P_{XY}$. Fix any $Q_{XYW}\in \calP(P_{XY})$. Let $Q_{\tilde X|XW}$ and $Q_{\tilde Y|YW}$ be the optimizers for the conditional RDP function (cf. \eqref{conditional RDP}). Define the joint distribution $Q_{XYW\tilde{X}\tilde{Y}} := P_{XY}Q_{W|XY}Q_{\tilde X|XW}Q_{\tilde Y|YW}$. Fix an integer $n\in\bbN$. Let $(M_0,M_1,M_2)\in \bbN^3$ be three integers to be specified later. Assume that the reproduction alphabets satisfy $\tilde{\mathcal{X}} \subseteq \mathcal{X}$, $\tilde{\mathcal{Y}} \subseteq \mathcal{Y}$ with $|\tilde{\mathcal{X}}| < \infty$ and $|\tilde{\mathcal{Y}}| < \infty$. 
By Assumption~\ref{assump:1}, for any $\varepsilon \in \bigl(0,\min\{\tfrac{D_1}{2}, \tfrac{D_2}{2},\tfrac{P_1}{2},\tfrac{P_2}{2}\}\bigr)$, we can obtain
\begin{align}
I(X;\tilde{X}|W)&\leq R_{X|W}(Q_{XW},D_1-2\varepsilon,P_1-2\varepsilon)+\varepsilon,\label{assump1_1}\\
I(Y;\tilde{Y}|W)&\leq R_{Y|W}(Q_{YW},D_2-2\varepsilon,P_2-2\varepsilon)+\varepsilon,\\
\bbE_{Q_{X\tilde X}}[\Delta_1(X,\tilde{X})] &\le D_1 - \varepsilon,\label{assump1_2}\\
\bbE_{Q_{Y\tilde Y}}[\Delta_2(Y,\tilde{Y})] &\le D_2 - \varepsilon,\\
d_1(P_X,Q_{\tilde{X}})&\leq P_1 - \varepsilon,\label{assump1_3}\\
d_2(P_Y,Q_{\tilde{Y}})&\leq P_2 - \varepsilon.
\end{align}
\subsection{Achievability Proof}
To establish the achievability part of Theorem~\ref{theorem:1}, we first describe the random codebook generation. Next, we specify the encoding and decoding schemes. The proof then proceeds by analyzing the error probability to guarantee the existence of a valid common codeword that is jointly typical with the source sequence. Conditioned on this common description, we verify that private codewords can be found to satisfy the distortion constraints. Subsequently, we demonstrate that the random circular shift mechanism aligns the reconstruction distribution with the source statistics to meet the perception constraints. Finally, we discuss the removal of common randomness.

Let $K$ be uniformly distributed over $[0 : n-1]$ and independent of the source sequence $(X^n,Y^n)$, and known to both the encoders and decoders. Fix a common randomness realization $k\in[0:n-1]$, and let $(\bar{X}^n,\bar{Y}^n):=\pi_{-k}(X^n,Y^n)$ denote the circularly shifted source sequences. Note that due to the i.i.d. property of the source, $(\bar{X}^n,\bar{Y}^n)$ follows the same distribution as $(X^n,Y^n)$. We define the codebooks, encoding, and decoding functions based on these shifted sequences as follows:\\

\begin{figure*}[t]
\centering
\resizebox{\linewidth}{!}{%
\begin{tikzpicture}[
>=Latex,
node distance=12mm and 16mm,
box/.style={draw, rounded corners, align=center, inner sep=4pt},
op/.style={draw, circle, align=center, inner sep=1.5pt},
lab/.style={align=center}
]

\node[box] (src) {$(X^n,Y^n)=(x^n,y^n)$};
\coordinate (bus) at ($(src.east)+(9mm,0)$);
\node[box, right=15mm of src] (f0) {$f_0$\\choose $s_0$ s.t.\\$(\bar{x}^n,\bar{y}^n,w^n(s_0))\in \calT_\delta^{(n)}(Q_{XYW})$};
\node[box, above=9mm of f0 ] (f1) {$f_1$\\chooses $s_1$ s.t.~\eqref{Encoder:f1}};
\node[box, below=9mm of f0] (f2) {$f_2$\\chooses $s_2$ s.t.~\eqref{Encoder:f2}};
\draw(src.east) -- node[midway, above] {$\pi_{-k}$}  (bus);
\draw[->] (bus) |- (f0.west);
\draw[->] (bus) |- (f1.west);
\draw[->] (bus) |- (f2.west);
\node[box, right=14mm of f1] (phi1) {$\phi_1$\\outputs $\tilde x^n(s_0,s_1)$};
\node[box, right=14mm of f2] (phi2) {$\phi_2$\\outputs $\tilde y^n(s_0,s_2)$};
\coordinate (bus1) at (f0 -| phi1);
\draw (f0) --node[midway, above] {$s_{0}$} (bus1);
\draw[->] (bus1) --(phi1);
\draw[->] (bus1) --(phi2);
\draw[->] (f1) --node[midway, above] {$s_{1}$} (phi1);
\draw[->] (f2) --node[midway, above] {$s_{2}$} (phi2);
\node[box, right=10mm of phi1] (out1) {Reconstruction\\$\hat x^n=\pi_k(\tilde x^n(s_0,s_1))$};
\node[box, right=10mm of phi2] (out2) {Reconstruction\\$\hat y^n=\pi_k(\tilde y^n(s_0,s_2))$};
\draw[->] (phi1) --node[midway, above] {$\pi_{k}$} (out1);
\draw[->] (phi2) --node[midway, above] {$\pi_{k}$} (out2);
\end{tikzpicture}%
}
\caption{Illustration of the encoding process with common randomness.}
\label{fig:encoding-shift}
\end{figure*}
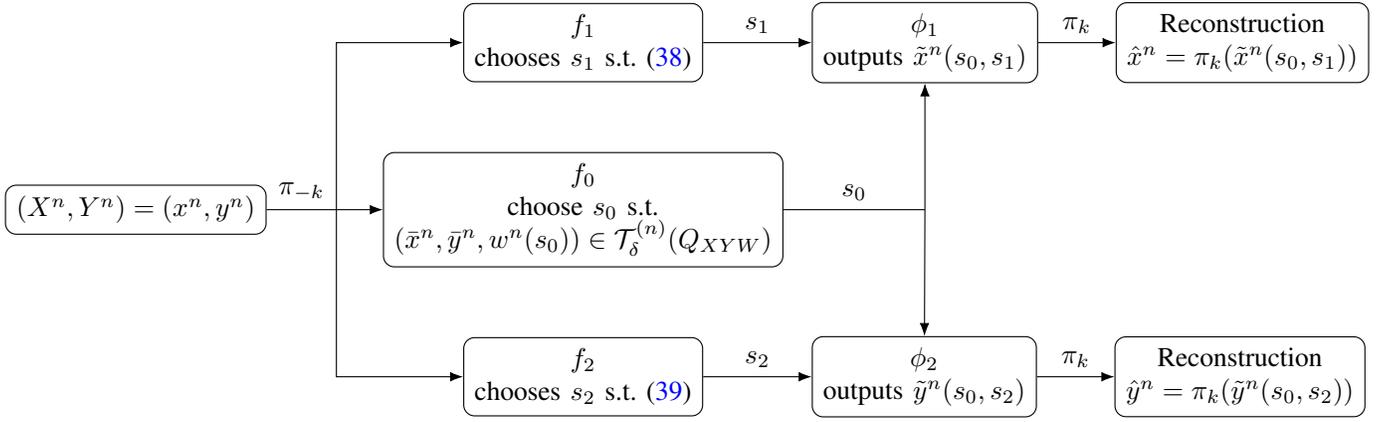

\subsubsection{Codebook Generation}
Fix any positive real number $\delta\in \bbR_+$, and define $\delta_1:=\delta\bigl(H(W) + H(W|XY)\bigr)$, $\delta^{\prime}:=\delta(H(\tilde{X}|W)+1)$, and $\delta^{\prime\prime}:=\delta(H(\tilde{Y}|W)+1)$. Firstly, we generate $M_0:=\lfloor\exp\left(n\left(I(X,Y;W) + 2\delta_1\right)\right)\rfloor$ public codewords $\bw:=\{w^n(1), \ldots, w^n(M_0)\}$ independently from a uniform distribution over the typical set $\calT_{\delta}^{(n)}(Q_W)$. Subsequently, for each $i\in[M_0]$, we generate $M_1:=\lfloor \exp({n ( I(X;\tilde X|W) + 2\delta^{\prime})}) \rfloor$ private codewords $\{\tilde{x}^n(i,1),\ldots,\tilde{x}^n(i,M_1)\}$ independently from a uniform distribution over the conditional typical set $\calT_{\delta}^{(n)}(Q_{\tilde X|W}|w^n(i))$ and generate $M_2:=\lfloor \exp({n ( I(Y;\tilde Y|W) + 2\delta^{\prime\prime} )}) \rfloor$ private codewords $\{\tilde{y}^n(i,1),\ldots,\tilde{y}^n(i,M_2)\}$ independently from a uniform distribution over the conditional typical set $\calT_{\delta}^{(n)}(Q_{\tilde Y|W}|w^n(i))$.

Here, the typical sets are defined as follows: $\calT_{\delta}^{(n)}(Q_W)$ denotes the set of $\delta$-typical sequences with respect to $Q_W$, and for any $i\in[M_0]$, $\calT_{\delta}^{(n)}(Q_{\tilde{X}|W}|w^n(i))$ and $\calT_{\delta}^{(n)}(Q_{\tilde{Y}|W}|w^n(i))$ denote the sets of conditional $\delta$-typical sequences with respect to $Q_{\tilde{X}|W}$, $Q_{\tilde{Y}|W}$ respectively. Specifically,
\begin{align}
\calT_{\delta}^{(n)}(Q_W) := \left\{ w^n \in {\mathcal{W}}^n : | \hat{T}_{w^n}(a) - Q_{W}(a) | \leq \delta Q_{W}(a),\,\forall~{a}\in\calW\right\},
\end{align}
\begin{align}
\calT_{\delta}^{(n)}(Q_{\tilde{X}|W}|w^n(i)) := \left\{ \tilde{x}^n \in \tilde{\mathcal{X}}^n : | \hat{T}_{\tilde{x}^n,w^n(i)}(a,b) - Q_{\tilde{X}W}(a,b) | \leq \delta Q_{\tilde{X}W}(a,b),\,\forall~{(a,b)\in\cal\tilde{X}\times\calW}\right\},\label{conditional_typeX}
\end{align}
\begin{align}
\calT_{\delta}^{(n)}(Q_{\tilde{Y}|W}|w^n(i)) := \left\{ \tilde{y}^n \in \tilde{\mathcal{Y}}^n : | \hat{T}_{\tilde{y}^n,w^n(i)}(a,b) - Q_{\tilde{Y}W}(a,b) | \leq \delta Q_{\tilde{Y}W}(a,b),\,\forall~{(a,b)\in\cal\tilde{Y}\times\calW}\right\}.
\end{align}\\
Thus, we have the codebooks $\tilde{\bx}=\{\tilde{x}^n(i,j)\}_{i\in[M_0],j\in[M_1]}$ and $\tilde{\by}=\{\tilde{y}^n(i,j)\}_{i\in[M_0],j\in[M_2]}$ for private messages. The three codebooks $(\bw,\tilde{\bx},\tilde{\by})$ will be used to construct our coding scheme.\\
\subsubsection{Coding Scheme}
Assume that the source sequence realization is $(x^n,y^n)\in \calX^n\times\calY^n$. Before compression, the encoders first compute the circularly shifted versions of the source sequences as $(\bar{x}^n,\bar{y}^n)=\pi_{-k}(x^n,y^n)$. Subsequently, the encoder $f_0$ selects an index $s_0\in[M_0]$ such that $(\bar{x}^n,\bar{y}^n,w^n(s_0))\in \calT_{\delta}^{(n)}(Q_{XYW})$. If there is no such index, $s_0$ is set to $1$; if there is more than one such index, $s_0$ is chosen as the smallest among all. Given $s_0\in[M_0]$, encoders $f_1$ and $f_2$ select indices $s_1\in[M_1]$ and $s_2\in[M_2]$ such that
\begin{align}
\frac{1}{n}\sum_{t\in[n]}\Delta_1(\bar{x}_t,\tilde{x}_t(s_0,s_1)) &\leq \bbE_{Q_{X\tilde{X}}}[\Delta_1(X,\tilde{X})] + \frac{\delta}{2},
\label{Encoder:f1}\\
\frac{1}{n}\sum_{t\in[n]}\Delta_2(\bar{y}_t,\tilde{y}_t(s_0,s_2)) &\leq \bbE_{Q_{Y\tilde{Y}}}[\Delta_2(Y,\tilde{Y})] + \frac{\delta}{2}. \label{Encoder:f2}
\end{align}
Again, if there is no such index, the index is set to $1$ and if there is more than one index, it is chosen as the smallest among all.

Upon receiving $(s_0,s_1)$, the decoder $\phi_1$ finds the codeword $\tilde{x}^n(s_0,s_1)$ and outputs the reproduced version as $\hat{x}^n=\pi_{k}(\tilde{x}^n(s_0,s_1))$ using the circular shift operator. Similarly, upon receiving $(s_0,s_2)$, the decoder $\phi_2$ finds the codeword $\tilde{y}^n(s_0,s_2)$ and outputs the reproduced version as $\hat{y}^n=\pi_{k}(\tilde{y}^n(s_0,s_2))$ using the circular shift operator.

In subsequent analyses, we use the random coding idea and thus consider randomness of source sequences and codewords. Thus, we consider random source sequences $(X^n,Y^n)$ and random codebooks $\bW=\{W^n(i)\}_{i\in[M_0]}$, $\tilde{\bX}=\{\tilX^n(i,j)\}_{i\in[M_0],j\in[M_1]}$ and $\tilde{\bY}=\{\tilY^n(i,j)\}_{i\in[M_0],j\in[M_2]}$, whose joint distribution satisfies that for each $(x^n,y^n,\bw,\tilde{\bx},\tilde{\by})$
\begin{align}
&P_{X^nY^n\bW\tilde{\bX}\tilde{\bY}}(x^n,y^n,\bw,\tilde{\bx},\tilde{\by})\nn\\
&=P_{XY}^n(x^n,y^n)\prod_{i\in[M_0]}\bigg(\rmU_{\calT_{\delta}^{n}(Q_W)}(w^n(i))
\times\prod_{j\in[M_1]}\rmU_{\calT_{\delta}^{n}(Q_{\tilX|W}|w^n(i))}(\tilx^n(i,j))
\times\prod_{j\in[M_2]}\rmU_{\calT_{\delta}^{n}(Q_{\tilY|W}|w^n(i))}(\tily^n(i,j))\bigg).
\end{align}
In what follows, when we use expectation or probability, it is calculated with respect to the above joint distribution or its induced (conditional) distributions.\\
\subsubsection{Distortion Analysis}
Define the event
\begin{align}
\mathcal{E}_0:= \{(\Bar{X}^n,\Bar{Y}^n,W^n(s_0)) \notin {\calT_{\delta}}^{(n)}(Q_{XYW}),\,\forall~{s_0}\in [M_0] \}.
\end{align}
By the law of total probability, we have
\begin{align}
\Pr\{\mathcal{E}_0\} 
&=\Pr_{P_{XY}^n\rmU_{\calT_{\delta}^{n}(Q_W)}}{\left\{\mathcal{E}_0,(\bar{X}^n,\Bar{Y}^n) \notin {\calT_{\delta}}^{(n)}(P_{XY})\right\}} +  \Pr_{P_{XY}^n\rmU_{\calT_{\delta}^{n}(Q_W)}}\left\{\mathcal{E}_0,(\Bar{X}^n,\Bar{Y}^n) \in {\calT_{\delta}}^{(n)}(P_{XY})\right\}\\
&\leq \Pr_{P_{XY}^n}{\left\{(\Bar{X}^n,\Bar{Y}^n) \notin {\calT_{\delta}}^{(n)}(P_{XY})\right\}} +  \Pr_{P_{XY}^n\rmU_{\calT_{\delta}^{n}(Q_W)}}\left\{\mathcal{E}_0,\,(\Bar{X}^n,\Bar{Y}^n) \in {\calT_{\delta}}^{(n)}(P_{XY})\right\}.
\end{align}
Since $(X^n,Y^n)$ is an i.i.d. sequence with distribution $P_{XY}$, the shifted sequence $(\Bar{X}^n,\Bar{Y}^n)$ has the same distribution as $(X^n,Y^n)$. Thus, by the properties of typical sets~\cite[Chapter 2]{Gamal_2011_Network-Information-Theory}, for sufficiently large $n$,
\begin{align}\label{source_joint}
\Pr_{P_{XY}^n}{\left\{(\Bar{X}^n,\Bar{Y}^n) \notin {\calT_{\delta}}^{(n)}(P_{XY})\right\}}
&\leq \varepsilon_1,
\end{align}
where $\varepsilon_1\to0$ as $n\to \infty$.
Then we can obtain
\begin{align}
&\Pr_{P_{XY}^n\rmU_{\calT_{\delta}^{n}(Q_W)}}\left\{
\mathcal{E}_0
,\,
(\Bar{X}^n,\Bar{Y}^n) \in \mathcal{T}_{\delta}^{(n)}(P_{XY})
\right\}\nn\\
&=\sum_{(\Bar{x}^n,\Bar{y}^n)\in \mathcal{T}_{\delta}^{(n)}(P_{XY})}
P_{XY}^n(\Bar{x}^n,\Bar{y}^n)
\,
\Pr_{\rmU_{\calT_{\delta}^{n}(Q_W)}}\Bigl\{
(\Bar{x}^n,\Bar{y}^n,W^n(s_0))
\notin \mathcal{T}_{\delta}^{(n)}(Q_{XYW}),\,\forall~{s_0}\in [M_0]
\Bigr\} \label{e1_1}
\\[0.5em]
&=
\sum_{(\Bar{x}^n,\Bar{y}^n)\in \mathcal{T}_{\delta}^{(n)}(P_{XY})}
P_{XY}^n(\Bar{x}^n,\Bar{y}^n)
\Bigl(
1-\sum_{w^n\in\mathcal{T}_{\delta}^{(n)}(Q_{W})} \rmU_{\calT_{\delta}^{n}(Q_W)}(w^n)\mathbf{1}\{(\Bar{x}^n,\Bar{y}^n,w^n)\in\mathcal{T}_{\delta}^{(n)}(Q_{XYW})\}
\Bigr)^{M_0}\label{e1_2}
\\[0.5em]
&=\sum_{(\Bar{x}^n,\Bar{y}^n)\in \mathcal{T}_{\delta}^{(n)}(P_{XY})}
P_{XY}^n(\Bar{x}^n,\Bar{y}^n)
\Bigl(
1 -\frac{\sum_{w^n\in\mathcal{T}_{\delta}^{(n)}(Q_{W})}\mathbf{1}\{(\Bar{x}^n,\Bar{y}^n,w^n)\in\mathcal{T}_{\delta}^{(n)}(Q_{XYW})\}}{|\mathcal{T}_{\delta}^{(n)}(Q_W)|}
\Bigr)^{M_0}\label{e1_3}
\\[0.5em]
&=
\sum_{(\Bar{x}^n,\Bar{y}^n)\in \mathcal{T}_{\delta}^{(n)}(P_{XY})}
P_{XY}^n(\Bar{x}^n,\Bar{y}^n)
\Bigl(
1 -
\frac{
\bigl|
\mathcal{T}_{\delta}^{(n)}(Q_{W|XY}\mid (\Bar{x}^n,\Bar{y}^n))
\bigr|
}{
\bigl|
\mathcal{T}_{\delta}^{(n)}(Q_W)
\bigr|
}
\Bigr)^{M_0}\label{e1_4}
\\[0.5em]
&\le
\sum_{(\Bar{x}^n,\Bar{y}^n)\in \mathcal{T}_{\delta}^{(n)}(P_{XY})}
P_{XY}^n(\Bar{x}^n,\Bar{y}^n)
\biggl(
1 -
\frac{
(1-\delta)\exp\!\Bigl(n\bigl(H(W|XY)-\delta H(W|XY)\bigr)\Bigr)
}{
\exp\!\bigl(n\bigl(H(W)+\delta H(W)\bigr)\bigr)
}
\biggr)^{M_0}\label{e1_5}
\\[0.5em]
&\leq
\Bigl(
1 -
(1-\delta)
\exp\!\left(
-n\bigl(
I(X,Y;W) + \delta(H(W)+H(W|XY))
\bigr)
\right)
\Bigr)^{M_0}\label{e1_6}
\\[0.5em]
&\le
\exp\!\Bigl(
-M_0(1-\delta)
\exp\!\left(
-n\bigl(
I(X,Y;W) + \delta(H(W)+H(W|XY))
\bigr)
\right)
\Bigr)\label{e1_7}
\\[0.5em]
&\leq\exp\!\left(
-(1-\delta)
\exp\Bigl(
n\delta \bigl(H(W)+H(W|XY)\bigr)
\Bigr)
\right)\label{e1_8}
\\[0.5em]
&\le \varepsilon_1,\label{e1_9}
\end{align}
where~\eqref{e1_1} follows from the total probability formula, \eqref{e1_2} holds because we independently and uniformly choose $w^n$ from $\calT^n_{\delta}(Q_W)$, \eqref{e1_3} holds because $\rmU_{\calT_{\delta}^{n}(Q_W)}$ is the uniform distribution over $\calT^n_{\delta}(Q_W)$, \eqref{e1_4} holds because $\mathbf{1}\{(\Bar{x}^n,\Bar{y}^n,w^n)\in\mathcal{T}_{\delta}^{(n)}(Q_{XYW})\} = 1$ only if $w^n \in {\calT_{\delta}}^{(n)}(Q_{W|XY}|(\Bar{x}^n,\Bar{y}^n))$, \eqref{e1_5} holds because $|\calT^n_{\delta}(Q_W)|\leq \exp(n(H(W)+\delta H(W)))$ and $|{\calT_{\delta}}^{(n)}(Q_{W|XY}|(\Bar{x}^n,\Bar{y}^n))|\geq (1-\delta)\exp(n(H(W|XY)-\delta H(W|XY)))$~\cite[Chapter 2]{Gamal_2011_Network-Information-Theory}, \eqref{e1_6} holds because $\sum_{(\Bar{x}^n,\Bar{y}^n)\in \mathcal{T}_{\delta}^{(n)}(P_{XY})}
P_{XY}^n(\Bar{x}^n,\Bar{y}^n) \leq 1$, \eqref{e1_7} follows from~\cite[Lemma 10.5.3]{Thomas_2006_Elements-of-information-theory}, i.e., for $0\leq a,b \leq 1$, $c>0$, $(1-ab)^c\leq1-a+\exp(-bc)$, \eqref{e1_8} holds because $M_0=\lfloor\exp\left(n\left(I(X,Y;W) + 2\delta_1\right)\right)\rfloor$, and~\eqref{e1_9} holds because $-(1-\delta)\exp(n\delta\left(H(W)+H(W|XY)\right)) $ tends to $-\infty$ as $n\to\infty$, which implies the exponential term is bounded by $\varepsilon_1$.\\
Combining~\eqref{source_joint} and~\eqref{e1_9}, we obtain
\begin{align}
\Pr\{\mathcal{E}_0\}\leq 2\varepsilon_1.\label{e_1_result}
\end{align}
Thus, the encoder $f_0$ can proceed successfully with high probability asymptotically.

Next, we analyze the distortion constraint for the source sequence $X^n$ since the analysis for $Y^n$ is analogous. Recall that the public codewords are generated from the typical set $\calT_{\delta}^{(n)}(Q_{W})$. Fix an index $s_0\in[M_0]$. Consider a source sequence $\bar{x}^n\in\calX^n$ and a public codeword $w^n(s_0)\in\calT_{\delta}^{(n)}(Q_{W})$.
Let $\hat{R}_1 := I(X;\tilde{X}|W) + \delta^{\prime}$. We define the sets
\begin{align}\label{defition_calA}
\mathcal{A}_{(\Bar{x}^n,w^n(s_0))} := \left\{ \tilde{x}^n \in \calT_{\delta}^{(n)}(Q_{\tilde{X}|W}|w^n(s_0)) : 
\begin{aligned} 
&Q_{\tilde{X}|XW}^n(\tilde{x}^n|\Bar{x}^n,w^n(s_0)) > \frac{\exp({n\hat{R}_1)}}{\lvert\calT_{\delta}^{(n)}(Q_{\tilde{X}|W}|w^n(s_0))\rvert} \\
& \quad \text{or} \quad
\frac{1}{n}\sum_{t=1}^{n}\Delta_1(\Bar{x}_t,\tilde{x}_{t}) > \bbE_{Q_{X\tilde X}}[\Delta_1(X,\tilde{X})] + \frac{\delta}{2} 
\end{aligned}
\right\},
\end{align}
and
\begin{align}
\mathcal{B}_{\Bar{x}^n} := \left\{ w^n \in \calT_{\delta}^{(n)}(Q_{W}): Q_{W|X}^n(w^n|\Bar{x}^n) \geq \frac{1}{\lvert\calT_{\delta}^{(n)}(Q_{W})\rvert} \right\}.\label{definition_B}
\end{align}
Since the coding scheme ensures that the pair $(\bar{x}^n,w^n(s_0))$ is jointly typical with high probability, it follows from the properties of typical sequences that $w^n(s_0)\in\mathcal{B}_{\Bar{x}^n}$ is satisfied.
\begin{lemma}\label{lemma:1}
Let $Q_{XW}$ be a joint probability distribution defined on finite alphabets $\mathcal{X}\times \mathcal{W}$ such that the mutual information satisfies $I(X;W) > 0$. For any $\delta>0$ sufficiently small, if a pair of sequences $(\Bar{x}^n,w^n)\in \calT_{\delta}^{(n)}(Q_{XW})$, we have $w^n\in\calB_{\bar{x}^n}$.
\end{lemma}
\begin{proof}
Let $\delta_{Q_W}:=\delta H(W) +  \frac{1}{n}\log\frac{1}{(1-\delta)}$. Invoking the property of the 
$\delta$-typical set $\calT_{\delta}^{(n)}(Q_{W})$ \cite[Chapter 2]{Gamal_2011_Network-Information-Theory}, we obtain
\begin{align}
|\calT_{\delta}^{(n)}(Q_{W})|\geq (1-\delta)\,
\exp\left(
n\bigl(H(W)-\delta H(W)\bigr)
\right).\label{typical_set_upper}
\end{align}
Thus, we have
\begin{align}
\frac{1}{n} \log \frac{1}{|\calT_{\delta}^{(n)}(Q_{W})|} 
&\leq \frac{1}{n}
\log\biggl(
\frac{
1
}{
(1-\delta)\,
\exp\!\left(
n\bigl(H(W)-\delta H(W)\bigr)
\right)
}
\biggr)\label{upper_bound_calB_1}\\ 
&=-H(W)+\delta H(W) +  \frac{1}{n}\log\frac{1}{(1-\delta)}\\
&=-H(W) + \delta_{Q_W},\label{B_right}
\end{align}
where \eqref{upper_bound_calB_1} follows from \eqref{typical_set_upper}. Note that $\delta_{Q_W}\to0$ as $\delta\to 0$. 
Given $(\Bar{x}^n,w^n)\in \calT_{\delta}^{(n)}(Q_{XW})$, we obtain
\begin{align}
\frac{1}{n} \log Q_{W|X}^n(w^n|\Bar{x}^n) 
&= \sum_{(a,b)\in \calX\times\calW} \hat{T}_{\Bar{x}^n,w^n}(a,b) \log Q_{W|X}(b|a)\label{left_B_1} \\
&\geq \sum_{(a,b)\in \calX\times\calW} (1+\delta) Q_{XW}(a,b) \log Q_{W|X}(b|a)\label{left_B_2}\\
&= (1+\delta){\sum_{(a,b)\in \calX\times\calW} Q_{XW}(a,b) \log Q_{W|X}(b|a)}\label{left_B_3}\\
&= -(1+\delta) H(W|X),\label{left_B_4}
\end{align}
where \eqref{left_B_1} is due to the joint empirical distribution $\hat{T}_{\Bar{x}^n,w^n}$, \eqref{left_B_2} follows from the definition of $\calT_{\delta}^{(n)}(Q_{XW})$\cite[Chapter 2]{Gamal_2011_Network-Information-Theory}, and \eqref{left_B_4} uses the definition of $H(W|X)$.
Since $I(X;W)>0$ and $\delta\to0$, we have
\begin{align}
\frac{1}{n} \log Q_{W|X}^n(w^n|\Bar{x}^n) - \frac{1}{n} \log \frac{1}{|\calT_{\delta}^{(n)}(Q_{W})|} 
&\geq  -(1+\delta) H(W|X) + H(W) -\delta_{Q_W}\label{diff_B}\\
&= I(X;W) - \delta H(W|X) -\delta_{Q_W}\\
&>0,
\end{align}
where \eqref{diff_B} follows from \eqref{B_right} and \eqref{left_B_4}.
\end{proof}
Fix a source sequence $x^n\in\calX^n$ and public codeword $w^n(s_0)\in\calT_{\delta}^{(n)}(Q_{W})$. Recall that $\Bar{x}^n=\pi_{-k}(x^n)$ and that the private codewords $\{\tilde X^n(s_0,s_1)\}_{s_1\in[M_1]}$ are generated independently from a uniform distribution over $\calT_{\delta}^{(n)}(Q_{\tilde{X}|W}|w^n(s_0))$. Define the event
\begin{align}
\mathcal{E}_1 :=\Bigl\{\frac{1}{n}\sum_{t\in[n]}\Delta_1(\Bar{x}_t,\tilde X_t(s_0,s_1))>\bbE_{Q_{X\tilde X}}[\Delta_1(X,\tilde X)]+\tfrac{\delta}{2},\,\forall~{s_1}\in [M_1]\Bigr\}.
\end{align} 
Therefore, we obtain
\begin{align}
\Pr\{\mathcal{E}_1\}
&= \prod_{s_1\in [M_1]}\Pr_{\rmU_{\calT_{\delta}^{(n)}(Q_{\tilde{X}|W}|w^n(s_0))}}\Bigl\{\tfrac{1}{n}\sum_{t\in[n]}\Delta_1(\Bar{x}_t,\tilde X_t(s_0,s_1))
>\bbE_{Q_{X\tilde X}}[\Delta_1(X,\tilde X)]+\tfrac{\delta}{2}\;\Bigm|\bar{x}^n,w^n(s_0)\Bigr\}\label{e0_1}\\
&\leq \prod_{s_1\in [M_1]}\sum_{\tilde x^n\in\mathcal{A}_{(\Bar{x}^n,w^n(s_0))}}\rmU_{\calT_{\delta}^{(n)}(Q_{\tilde{X}|W}|w^n(s_0))}(\tilde x^n) \label{e0_2}\\
&=\biggl(\sum_{\tilde x^n\in\mathcal{A}_{(\Bar{x}^n,w^n(s_0))}}\rmU_{\calT_{\delta}^{(n)}(Q_{\tilde{X}|W}|w^n(s_0))}(\tilde x^n)\biggr)^{M_1}\label{e0_3}\\
&= \Bigl(1-\sum_{\tilde x^n\in \calT_{\delta}^{(n)}(Q_{\tilde{X}|W}|w^n(s_0))\setminus\mathcal{A}_{(\Bar{x}^n,w^n(s_0))}}\rmU_{\calT_{\delta}^{(n)}(Q_{\tilde{X}|W}|w^n(s_0))}(\tilde x^n)\Bigr)^{M_1}\label{e0_4}\\
&= \Bigl(1-\sum_{\tilde x^n\in\calT_{\delta}^{(n)}(Q_{\tilde{X}|W}|w^n(s_0))\setminus\mathcal{A}_{(\Bar{x}^n,w^n(s_0))}}
\tfrac1{|\calT_{\delta}^{(n)}(Q_{\tilde{X}|W}|w^n(s_0))|}\Bigr)^{M_1}\label{e0_5}\\
&\le \Bigl(1-\exp({-n\hat R_1})\sum_{\tilde x^n\in\calT_{\delta}^{(n)}(Q_{\tilde{X}|W}|w^n(s_0))\setminus\mathcal{A}_{(\Bar{x}^n,w^n(s_0))}}
Q_{\tilde X|XW}^n(\tilde x^n|\Bar{x}^n,w^n(s_0))\Bigr)^{M_1}\label{e0_6}\\ 
&\le 1-\sum_{\tilde x^n\in\calT_{\delta}^{(n)}(Q_{\tilde{X}|W}|w^n(s_0))\setminus\mathcal{A}_{(\Bar{x}^n,w^n(s_0))}}
Q_{\tilde X|XW}^n(\tilde x^n|\Bar{x}^n,w^n(s_0))
+\exp\bigl(-M_{1} \exp(-n\hat{R}_1)\bigr)\label{e0_7}\\ 
&= \sum_{\tilde x^n\in(\calT_{\delta}^{(n)}(Q_{\tilde{X}|W}|w^n(s_0)))^{\rmc}\cup\mathcal{A}_{(\Bar{x}^n,w^n(s_0))}}
Q_{\tilde X|XW}^n(\tilde x^n|\Bar{x}^n,w^n(s_0))
+\exp\bigl(-M_{1} \exp(-n\hat{R}_1)\bigr),\label{e0_8}
\end{align}
where \eqref{e0_1} holds because the codewords for different indices are generated independently, 
\eqref{e0_2} follows from the definition of $\mathcal{A}_{(\Bar{x}^n,w^n(s_0))}$ in \eqref{defition_calA}, which implies that any sequence $\tilde{x}^n$ satisfying the distortion condition in $\mathcal{E}_1$ must belong to $\mathcal{A}_{(\Bar{x}^n,w^n(s_0))}$. \eqref{e0_3} holds because the term inside the product is independent of the index $s_1$, \eqref{e0_5} holds because $\rmU_{\calT_{\delta}^{(n)}(Q_{\tilde{X}|W}|w^n(s_0))}(\tilde{x}^n)$ is the uniform distribution over $\calT_{\delta}^{(n)}(Q_{\tilde X|W}|w^n(s_0))$, \eqref{e0_6} follows by substituting the bound $\frac{1}{|\calT_{\delta}^{(n)}(Q_{\tilde{X}|W}|w^n(s_0))|}\geq Q_{\tilde X|XW}^n(\tilde x^n|\bar{x}^n,w^n(s_0))\exp(-n\hat{R}_1)$ which holds for any $\tilde{x}^n \notin \mathcal{A}_{(\bar{x}^n,w^n(s_0))}$, \eqref{e0_7} follows from~\cite[Lemma 10.5.3]{Thomas_2006_Elements-of-information-theory}, i.e., for $0\leq a,b \leq 1$, $c>0$, $(1-ab)^c\leq1-a+\exp(-bc)$, by setting $a=\sum_{\tilde x^n\in\calT_{\delta}^{(n)}(Q_{\tilde{X}|W}|w^n(s_0))\setminus\mathcal{A}_{(\Bar{x}^n,w^n(s_0))}}
Q_{\tilde X|XW}^n(\tilde x^n|\Bar{x}^n,w^n(s_0))$, $b=\exp(-n\hat{R}_1)$, and $c=M_1$, \eqref{e0_8} holds because $ (\calT_{\delta}^{(n)}(Q_{\tilde{X}|W}|w^n(s_0))\setminus\mathcal{A}_{(\Bar{x}^n,w^n(s_0))})^{\rmc}= (\calT_{\delta}^{(n)}(Q_{\tilde{X}|W}|w^n(s_0)))^{\rmc}\cup\mathcal{A}_{(\Bar{x}^n,w^n(s_0))}$.

Define
\begin{align}
{\mathcal{A}_1}^{(n)} := \left\{ (\Bar{x}^n, w^n(s_0), \tilde{x}^n) :(\Bar{x}^n,w^n(s_0))\in \calT_{\delta}^{(n)}(Q_{XW}), \tilde{x}^n \in (\calT_{\delta}^{(n)}(Q_{\tilde{X}|W}|w^n(s_0)))^{\rmc} \cup \mathcal{A}_{(\Bar{x}^n,w^n(s_0))} \right\}.
\end{align}
We can obtain
\begin{align}
&\Pr\biggl\{ \frac{1}{n} \sum_{t\in[n]} \Delta_1(\Bar{X}_t, \tilde{X}_t(s_0,s_1)) > \bbE_{Q_{X\tilde X}}[\Delta_1(X, \tilde{X})] + \frac{\delta}{2} \biggr\}\nonumber \\
&=\Pr\biggl\{ \frac{1}{n} \sum_{t\in[n]} \Delta_1(\Bar{X}_t, \tilde{X}_t(s_0,s_1)) > \bbE_{Q_{X\tilde X}}[\Delta_1(X, \tilde{X})] + \frac{\delta}{2}\,,(\Bar{X}^n,W^n(s_0))\in \calT_{\delta}^{(n)}(Q_{XW})\biggr\}\nonumber\\
&\quad+\Pr\biggl\{ \frac{1}{n} \sum_{t\in[n]} \Delta_1(\Bar{X}_t, \tilde{X}_t(s_0,s_1)) > \bbE_{Q_{X\tilde X}}[\Delta_1(X, \tilde{X})] + \frac{\delta}{2}\,, (\Bar{X}^n,W^n(s_0))\notin \calT_{\delta}^{(n)}(Q_{XW})\biggr\}\label{total_e0_1}\\
&\leq \Pr\biggl\{ \frac{1}{n} \sum_{t\in[n]} \Delta_1(\Bar{X}_t, \tilde{X}_t(s_0,s_1)) > \bbE_{Q_{X\tilde X}}[\Delta_1(X, \tilde{X})] + \frac{\delta}{2},\,(\Bar{X}^n,W^n(s_0))\in \calT_{\delta}^{(n)}(Q_{XW})\biggr\}\nn\\
&\quad+\Pr\biggl\{(\Bar{X}^n,W^n(s_0))\notin \calT_{\delta}^{(n)}(Q_{XW})\biggr\}\label{total_e0_2}\\
&\leq \sum_{(\Bar{x}^n,w^n(s_0))\in \calT_{\delta}^{(n)}(Q_{XW})}P_{X}^n(\Bar{x}^n)\rmU_{\calT_{\delta}^{(n)}(Q_{W})}(w^n(s_0))\Pr\{\mathcal{E}_1\}+2\varepsilon_1\label{total_e0_3}\\
&=\sum_{(\Bar{x}^n,w^n(s_0))\in \calT_{\delta}^{(n)}(Q_{XW})}P_{X}^n(\Bar{x}^n)\frac{1}{\lvert\calT_{\delta}^{(n)}(Q_{W})\rvert}\Pr\{\mathcal{E}_1\}+2\varepsilon_1\label{total_e0_4}\\
&\leq\sum_{(\Bar{x}^n,w^n(s_0))\in \calT_{\delta}^{(n)}(Q_{XW})}P_{X}^n(\Bar{x}^n)Q_{W|X}^n(w^n(s_0)|\bar{x}^n)\Pr\{\mathcal{E}_1\}+2\varepsilon_1\label{total_e0_5}\\
&\leq \sum_{(\Bar{x}^n,w^n(s_0))\in \calT_{\delta}^{(n)}(Q_{XW})}P_{X}^n(\Bar{x}^n)Q_{W|X}^n(w^n(s_0)|\Bar{x}^n)\sum_{\tilde x^n\in(\calT_{\delta}^{(n)}(Q_{\tilde{X}|W}|w^n(s_0)))^{\rmc}\cup\mathcal{A}_{(\Bar{x}^n,w^n(s_0))}}
Q_{\tilde X|XW}^n(\tilde x^n|\Bar{x}^n,w^n(s_0))\nn\\ 
&\quad+\exp\bigl( -M_1 \exp(-n \hat{R}_1)\bigr)+2\varepsilon_1\label{total_e0_6}\\
&= \sum_{(\Bar{x}^n,w^n(s_0),\tilde x^n) \in {\mathcal{A}_1}^{(n)}} Q_{XW\tilde X}^n(\Bar{x}^n,w^n(s_0),\tilde x^n) + \exp\bigl( -M_1 \exp(-n \hat{R}_1)\bigr) +2\varepsilon_1\label{total_e0_7}\\
&=\Pr_{Q_{XW\tilde{X}}^n}\left\{ (\Bar{X}^n, W^n(s_0),\check{X}^n) \in {\mathcal{A}}_1^{(n)}\right\}+ \exp\bigl( -M_1 \exp(-n \hat{R}_1)\bigr) +2\varepsilon_1\label{total_e0_8}.
\end{align}
Note that for $t\in[n]$, $(\Bar{X}_t, W_t(s_0),\check{X}_t)$ are independent and distributed according to \( Q_{XW\tilde{X}} \). \eqref{total_e0_1} follows from the law of total probability, \eqref{total_e0_3} holds because the codewords $w^n(s_0)$ are chosen independently and uniformly from $\calT_{\delta}^{(n)}(Q_W)$, and the second term is bounded by $2\varepsilon_1$ as shown in \eqref{e_1_result}, \eqref{total_e0_4} holds for the definition of $U_{\calT_{\delta}^{(n)}(Q_W)}(w^n(s_0))$, \eqref{total_e0_5} follows from Lemma~\ref{lemma:1}, which implies $\frac{1}{|\calT_{\delta}^{(n)}(Q_W)|}\leq Q_{W|X}^n(w^n(s_0)|\bar{x}^n)$ for typical pairs, \eqref{total_e0_6} follows from \eqref{e0_8}, \eqref{total_e0_7} follows from the definition of ${\mathcal{A}_1}^{(n)}$. 

According to the definition of ${\mathcal{A}_1}^{(n)}$, it follows that
\begin{align}
&\Pr_{Q_{XW\tilde{X}}^n}\big\{(\Bar{X}^n, W^n(s_0),\check{X}^n) \in {\mathcal{A}}_1^{(n)}\big\}\nn\\
&\leq \Pr_{Q_{\tilde{X}}^n}\left\{\check{X}^n \notin \calT_{\delta}^{(n)}(Q_{\tilde{X}|W}|W^n(s_0))\right\} + \Pr_{Q_{X\tilde{X}}^n}\biggl\{ \Bigl|\frac{1}{n} \sum_{t\in[n]} \Delta_1(\Bar{X}_t, \check{X}_t) - \bbE_{Q_{X\tilde X}}[\Delta_1(X, \tilde{X})] \Bigr|> \frac{\delta}{2} \biggr\}\nonumber \\
&\quad+ \Pr_{Q_{XW\tilde{X}}^n} \biggl\{\frac{1}{n}\log Q_{\tilde{X}|XW}^n(\check{X}^n|\Bar{X}^n,W^n(s_0)) > \frac{1}{n} \log \frac{\exp({n\hat{R}_1})}{|\calT_{\delta}^{(n)}(Q_{\tilde{X}|W}|W^n(s_0))|} \biggr\}.
\end{align}
Note that $\check{X}^n \in \calT_{\delta}^{(n)}(Q_{\tilde{X}|W}|W^n(s_0))$ with high probability for sufficiently large $n$, since the tuples $(\Bar{X}_t, W_t(s_0),\check{X}_t)$ for $t \in [n]$ are i.i.d. generated according to \( Q_{XW\tilde{X}} \). That is, there exists $\varepsilon_1\to0$ such that
\begin{align}\label{A_1}
\Pr_{Q_{\tilde{X}}^n}\{\check{X}^n \notin \calT_{\delta}^{(n)}(Q_{\tilde{X}|W}|W^n(s_0))\} \leq \varepsilon_1.
\end{align}
Moreover, assuming finite variances $\mathrm{Var}(\Delta_1(X, \tilde{X}))$ and $\mathrm{Var}(\log Q_{\tilde{X}|XW})$, Chebyshev’s inequality implies that for sufficiently large $n$:
\begin{align}\label{A_2}
\Pr_{Q_{X\tilde{X}}^n}\biggl\{ \Bigl|\frac{1}{n} \sum_{t\in[n]} \Delta_1(\Bar{X}_t, \check{X}_t) - \bbE_{Q_{X\tilde X}}[\Delta_1(X, \tilde{X})] \Bigr|> \frac{\delta}{2} \biggr\} \leq \frac{\mathrm{Var}(\Delta_1(X, \tilde{X}))}{n(\frac{\delta}{2})^2} \leq \varepsilon_1,
\end{align}
and
\begin{align}
\Pr_{Q_{XW\tilde{X}}^n} \left\{ \Bigl|\frac{1}{n}\log Q_{\tilde{X}|XW}^n(\check{X}^n |\Bar{X}^n,W^n(s_0)) + H(\tilde{X}|XW) \Bigr|\geq \delta\right\} \leq \frac{\mathrm{Var}(\log Q_{\tilde{X}|XW})}{n\delta^2}\leq \varepsilon_1.\label{Q_tilX|XW}
\end{align}
Using the property of the conditional
$\delta$-typical set $\calT_{\delta}^{(n)}(Q_{\tilde{X}|W}|W^n(s_0))$ \cite[Chapter 2]{Gamal_2011_Network-Information-Theory}, we have
\begin{align}
|\calT_{\delta}^{(n)}(Q_{\tilde{X}|W}|W^n(s_0))| 
&\leq \exp\bigl({n(1+\delta)H(\tilde X|W)}\bigr).\label{typical_set_Q_X|W}
\end{align}
Thus, we obtain
\begin{align}
\frac{1}{n} \log \frac{\exp({n\hat{R}_1})}{|\calT_{\delta}^{(n)}(Q_{\tilde{X}|W}|W^n(s_0))|}
&\geq \frac{1}{n} \log \frac{\exp({n\hat{R}_1})}{\exp\bigl({n(1+\delta)H(\tilde X|W)}\bigr)}\label{typical_set_Q_X|W_1}\\
&= I(X; \tilde{X}|W) + \delta^{\prime} - (1 + \delta) H(\tilde{X}|W)\label{typical_set_Q_X|W_2}\\
&= -H(\tilde{X}|XW) + \delta^{\prime} - \delta H(\tilde{X}|W)\label{typical_set_Q_X|W_3}\\
&= -H(\tilde{X}|XW) + \delta,\label{typical_set_Q_X|W_4}
\end{align}
where \eqref{typical_set_Q_X|W_1} follows from \eqref{typical_set_Q_X|W}, \eqref{typical_set_Q_X|W_2} follows from $\hat{R}_1 := I(X;\tilde{X}|W) + \delta^{\prime}$, \eqref{typical_set_Q_X|W_4} follows from $\delta^{\prime}=\delta(H(\tilde{X}|W)+1)$. Combining \eqref{Q_tilX|XW} and \eqref{typical_set_Q_X|W_4} yields
\begin{align}\label{A_3}
\Pr_{Q_{XW\tilde{X}}^n}\biggl\{\frac{1}{n}\log Q_{\tilde{X}|XW}^n(\check{X}^n |\Bar{X}^n,W^n(s_0)) > \frac{1}{n} \log \frac{\exp({n\hat{R}_1})}{|\calT_{\delta}^{(n)}(Q_{\tilde{X}|W}|W^n(s_0))|} \biggr\} \leq \varepsilon_1.
\end{align}
Finally, combining \eqref{A_1}, \eqref{A_2} and \eqref{A_3}, we have
\begin{align}
\Pr_{Q_{XW\tilde{X}}^n}\big\{(\bar{X}^n, W^n(s_0),\check{X}^n) \in {\mathcal{A}}_1^{(n)}\big\}
&\leq 3\varepsilon_1, \label{result_dist_1}
\end{align}
as \(n \to \infty\). 

Recall that $M_1 := \lfloor \exp(n(I(X;\tilde{X}|W) + 2\delta^{\prime})) \rfloor$. Observe that for sufficiently large $n$, 
\begin{align}
\exp\bigl( -M_1 \exp(-n \hat{R}_1) \bigr) &\leq \exp\Bigl(-\exp\bigl({n( I(X; \tilde{X}|W) + 2\delta' )}\bigr) \exp\bigl({-n(I(X;\tilde{X}|W) + \delta')}\bigr)\Bigr)\\
&=\exp\bigl(-\exp({n\delta')}\bigr)\\
&\leq \varepsilon_1,\label{result_dist_2}
\end{align}
where $\delta'\to0$ as $\delta\to0$.
Consequently, for any \(\delta > 0\) and setting \(\varepsilon_2 =6\varepsilon_1\), it follows that for sufficiently large \(n\),
\begin{align}
\frac{1}{n} \log M_1 \leq I(X; \tilde{X}|W) + 2\delta (H(\tilde{X}|W) + 1),\label{M1_code}
\end{align}
and 
\begin{align}\label{result_distortion}
\Pr\Bigl\{ \frac{1}{n} \sum_{t=1}^{n} \Delta_1(\Bar{X}_t, \tilde{X}_t) > \bbE_{Q_{X\tilde X}}[\Delta_1(X, \tilde{X})] + \frac{\delta}{2} \Bigr\} \leq \varepsilon_2,
\end{align}
where \eqref{result_distortion} follows from \eqref{total_e0_8}, \eqref{result_dist_1} and \eqref{result_dist_2}.
Define
\begin{align}\label{definition_V}
V := 
\begin{cases} 
1, & \frac{1}{n} \sum_{t=1}^{n} \Delta_1(\Bar{X}_t, \tilde{X}_t) > \bbE_{Q_{X\tilde X}}[\Delta_1(X, \tilde{X})] + \frac{\delta}{2}, \\
0, & \text{otherwise},
\end{cases}
\end{align}
and let \( Z_t := \max_{\tilde{x} \in \tilde{X}} \Delta_1(\Bar{X}_t, \tilde{x})  \) for \( t \in [n] \). We have
\begin{align}
\frac{1}{n} \sum_{t\in[n]} \bbE_{Q_{X_t\tilde X_t}}[\Delta_1(\Bar{X}_t, \tilde{X}_t)]
&= \frac{1}{n} \sum_{t\in[n]} \Pr\{V = 0\} \bbE_{Q_{X_t\tilde X_t}}[\Delta_1(\Bar{X}_t, \tilde{X}_t) \mid V = 0] \nn\\
&\quad+\frac{1}{n} \sum_{t\in[n]} \Pr\{V = 1\} \bbE_{Q_{X_t\tilde X_t}}[\Delta_1(\Bar{X}_t, \tilde{X}_t) \mid V = 1]\label{expectation_distortion_1}\\
&\leq \frac{1}{n} \sum_{t\in[n]} \bbE_{Q_{X_t\tilde X_t}}[\Delta_1(\Bar{X}_t, \tilde{X}_t) \mid V = 0] + \frac{1}{n} \sum_{t\in[n]} \bbE[V \Delta_1(\Bar{X}_t, \tilde{X}_t)]\label{expectation_distortion_2}\\
&\leq \bbE_{Q_{X\tilde X}}[\Delta_1(X, \tilde{X})] + \frac{\delta}{2} + \frac{1}{n} \sum_{t\in[n]} \bbE[V Z_t],\label{expectation_distortion_3}
\end{align}
where \eqref{expectation_distortion_1} follows from the law of total expectation, \eqref{expectation_distortion_2} follows from \eqref{result_distortion}, \eqref{expectation_distortion_3} follows from \eqref{definition_V} and the definition of $Z_t$.
For $t\in[n]$, we obtain
\begin{align}
\bbE[VZ_t] 
&= \Pr \Bigl\{ Z_t \leq \frac{1}{\sqrt{\varepsilon_2}} \Bigr\}\bbE \Bigl[ VZ_t \mid Z_t \leq \frac{1}{\sqrt{\varepsilon_2}} \Bigr] 
+ \Pr\Bigl\{ Z_t > \frac{1}{\sqrt{\varepsilon_2}} \Bigr\}\bbE \Bigl[ VZ_t \mid Z_t > \frac{1}{\sqrt{\varepsilon_2}} \Bigr]\label{VZ_1}\\
&\leq \Pr\Bigl\{ Z_t \leq \frac{1}{\sqrt{\varepsilon_2}} \Bigr\}\bbE \Bigl[ V \frac{1}{\sqrt{\varepsilon_2}} \mid Z_t \leq \frac{1}{\sqrt{\varepsilon_2}} \Bigr] 
+ \Pr\Bigl\{ Z_t > \frac{1}{\sqrt{\varepsilon_2}} \Bigr\}\bbE_{P_X} \Bigl[ Z_t \mid Z_t > \frac{1}{\sqrt{\varepsilon_2}} \Bigr]\label{VZ_2}\\
&\leq \frac{1}{\sqrt{\varepsilon_2}} \bbE[V] + \Pr\Bigl\{ Z_t > \frac{1}{\sqrt{\varepsilon_2}} \Bigr\}\bbE_{P_X} \Bigl[ Z_t \mid Z_t > \frac{1}{\sqrt{\varepsilon_2}} \Bigr]\label{VZ_3}\\
&\leq \sqrt{\varepsilon_2} + \Pr\Bigl\{ Z_t > \frac{1}{\sqrt{\varepsilon_2}} \Bigr\}\bbE_{P_X} \Bigl[ Z_t \mid Z_t > \frac{1}{\sqrt{\varepsilon_2}} \Bigr],\label{VZ_4}
\end{align}
where \eqref{VZ_1} follows from the law of total expectation, \eqref{VZ_2} holds because $Z_t\leq \frac{1}{\sqrt{\varepsilon_2}}$ and $V\leq 1$, \eqref{VZ_3} follows from the law of total expectation, \eqref{VZ_4} follows from \eqref{result_distortion}. 
Since for $t\in[n]$, we have
\begin{align}
\bbE_{P_X}[Z_t] = \bbE_{P_X} \Bigl[ \max_{\tilde{x} \in \mathcal{X}} \Delta_1(X, \tilde{x}) \Bigr] < \infty.
\end{align}
It follows by the dominated convergence theorem \cite[Theorem 16.4]{billingsley1995probability} that \(\Pr\{Z_t > z\} \bbE_{P_X}[Z_t \mid Z_t > z] \to 0\) as \(z \to \infty\). Therefore, by choosing a sufficiently small \(\varepsilon_2\), we can ensure that 
\begin{align}\label{distortion_3}
\bbE[VZ_t] \leq \frac{\delta}{2},\,t\in[n].
\end{align}
Combining \eqref{expectation_distortion_3} and \eqref{distortion_3}, and choosing $\delta\leq \frac{\varepsilon}{2(H(\tilde{X}|W)+1)}$, we have
\begin{align}
\frac{1}{n} \sum_{t\in[n]} \bbE_{Q_{X_t\tilde{X}_t}}[\Delta_1(\Bar{X}_t, \tilde{X}_t)] 
&\leq \bbE_{Q_{X\tilde{X}}}[\Delta_1(X, \tilde{X})] + \delta\\
&\leq \bbE_{Q_{X\tilde{X}}}[\Delta_1(X, \tilde{X})] + \frac{\varepsilon}{2(H(\tilde{X}|W)+1)}\\
&\leq \bbE_{Q_{X\tilde{X}}}[\Delta_1(X, \tilde{X})] + \frac{\varepsilon}{2}\\
&\leq D_1 - \varepsilon + \frac{\varepsilon}{2}\label{D1-e/2}\\
&= D_1 - \frac{\varepsilon}{2}\label{D1-e/2_1},
\end{align}
where \eqref{D1-e/2} follows from \eqref{assump1_2}.
Recalling that $\Bar{X}^n=\pi_{-k}(X^n)$ and the definition of the reconstruction sequence $\hat{X}^n = \pi_k{(\tilde X^n)}$ for $k\in[0:n-1]$, we obtain
\begin{align}
\frac{1}{n}\sum_{t\in[n]} \bbE_{Q_{X_t\hat{X}_t}}[\Delta_1(X_t,\hat{X}_{t})] 
& = \frac{1}{n}\sum_{t\in[n]} \Bigl(\sum_{k=0}^{n-1} \frac{1}{n}\bbE_{Q_{X_t\hat{X}_t}}[\Delta_1(X_t,\hat{X}_{t})|K = k] \Bigr)\label{distortion_measure_1}\\
& = \frac{1}{n}\sum_{t\in[n]}\bbE_{Q_{X_t\hat{X}_t}}[\Delta_1(X_t,\hat{X}_{t})|K = k]\label{distortion_measure_2}\\
&=\frac{1}{n}\sum_{t\in[n]} \bbE_{Q_{X_t\hat{X}_t}}[\Delta_1(X_t,\tilde{X}_{\theta_k^{(n)}(t)})|K = k]\label{distortion_measure_3}\\
&=\frac{1}{n}\sum_{t\in[n]} \bbE_{Q_{X_t\hat{X}_t}}[\Delta_1(\Bar{X}_{\theta_k^{(n)}(t)},\tilde{X}_{\theta_k^{(n)}(t)})|K = k]\label{distortion_measure_4} \\
&=\frac{1}{n}\sum_{i=1}^n \bbE_{Q_{X_i\tilde{X}_i}}[\Delta_1(\Bar{X}_{i},\tilde{X}_{i})|K = k] \label{distortion_measure_5}\\
&=\frac{1}{n}\sum_{i=1}^n\Bigl( \sum_{k=0}^{n-1}\frac{1}{n}\bbE_{Q_{X_i\tilde {X}_i}}[\Delta_1(\Bar{X}_{i},\tilde{X}_{i})|K = k] \Bigr)\label{distortion_measure_6}\\
&=\frac{1}{n}\sum_{t\in[n]} \bbE_{Q_{X_t\tilde {X}_t}}[\Delta_1(\Bar{X}_t,\tilde{X}_{t})]\label{distortion_measure_7}\\
&\le D_1 - \frac{\varepsilon}{2}\label{distortion_measure_8},
\end{align}
where \eqref{distortion_measure_1} follows from the law of total expectation, \eqref{distortion_measure_2} holds because the term $\sum_{t\in[n]}\bbE_{Q_{X_t\hat{X}_t}}[\Delta_1(X_t,\hat{X}_{t})|K = k]$ is independent of $k$ due to the stationarity of the source and the symmetry of the circular shift operation, \eqref{distortion_measure_3} follows from the construction $\hat{X}_{t}=\tilde{X}_{\theta_k^{(n)}(t)}$, \eqref{distortion_measure_4} follows from the relation $X_t=\bar{X}_{\theta_k^{(n)}(t)}$, \eqref{distortion_measure_5} holds because the formula is independent of the value of $k$, \eqref{distortion_measure_7} follows again from the law of total expectation, and \eqref{distortion_measure_8} follows from \eqref{D1-e/2_1}.\\
\subsubsection{Perception Analysis}
Again, we analyze the perception
constraint for the source sequence $X^n$ since the analysis
for $Y^n$ is analogous. Since the shift index $K$ is uniformly distributed on $[0:n-1]$, for any $t\in[n]$, $t' \in [n]$, we have
\begin{align}
Q_{\hat{X}_t|\tilde X_{t'}} = \frac{1}{n}.
\end{align}
Consequently, for any $t \in [n]$, the distribution of $\hat{X}_{t}$ satisfies
\begin{align}
Q_{\hat{X}_{t}} &= \sum_{t'=1}^n Q_{\hat{X}_{t}|\tilde{X}_{t'}}Q_{\tilde{X}_{t'}}\\ 
&= \frac{1}{n}\sum_{t'=1}^n Q_{\tilde{X}_{t'}}.\label{checkX} 
\end{align}
Now fix any public codeword $w^n \in \calT_{\delta}^{(n)}(Q_W)$. Since all private codewords are taken from the set $\calT_{\delta}^{(n)}(Q_{\tilde{X}|W}|w^n)$, i.e.,
\begin{align}
\Pr\{\tilde{X}^n \in \calT_{\delta}^{(n)}(Q_{\tilde{X}|W}|w^n)\} &= 1.\label{tilde X in Q_XW}
\end{align}
By the definition of empirical distribution, for a sequence $\tilde x^n\in \tilde{\calX}^n$ and any $a\in\cal\tilde{X}$, we have
\begin{align}
\hat{T}_{\tilde x^n}(a) &= \frac{1}{n}\sum_{t = 1}^n \mathbf{1}\{a = \tilde x_t\}\label{empirical distribution_1}\\ &= \frac{1}{n}\sum_{t = 1}^n Q_{\tilde X_t|\tilde X^n = \tilde x^n}(a),\label{empirical distribution_2}
\end{align}
where \eqref{empirical distribution_2} follows from the fact that, conditioned on $\tilde{X}^n=\tilde x^n$, the value at each position is determined.\\
For any $a\in \tilde{\calX}$, $b\in \calW$ and $\tilde{x}^n\in\calT_{\delta}^{(n)}(Q_{\tilde{X}|W}|w^n)$, we obtain
\begin{align}
\bigl|\hat{T}_{\tilde x^n}(a) - Q_{\tilde{X}}(a)\bigr|
&=\Bigl| \sum_{b\in \calW} \hat{T}_{\tilde x^n,w^n}(a,b)- \sum_{b\in \calW}Q_{\tilde XW}(a,b)\Bigr|\label{TV_1} \\
&\leq \sum_{b\in \calW}\left|\hat{T}_{\tilde x^n,w^n}(a,b)-Q_{\tilde XW}(a,b)\right|\label{TV_2}\\
&\leq \sum_{b\in \calW}\delta Q_{\tilde XW}(a,b)\label{TV_3}\\
&= \delta Q_{\tilde X}(a),\label{TV_4}
\end{align}
where \eqref{TV_1} follows from the law of total probability, \eqref{TV_2} follows from the triangle inequality, \eqref{TV_3} is due to the definition of the conditional typical set \eqref{conditional_typeX} and \eqref{tilde X in Q_XW}, and \eqref{TV_4} follows from the law of total probability.
Moreover, for $t\in[n]$, we demonstrate that $Q_{\hat{X}_t}$ converges to $Q_{\tilde{X}}$ in TV distance as $\delta \to 0$. Specifically, 
\begin{align}
&d_{\mathrm{TV}}\!\bigl( Q_{\hat{X}_t}, Q_{\tilde{X}}\bigr)\nn\\
&=d_{\mathrm{TV}}\!\Bigl( \frac{1}{n}\sum_{t\in[n]} Q_{\tilde{X}_t}, Q_{\tilde{X}}\Bigr)\label{TV_estimate_1}\\
&= \sum_{a\in \tilde \calX} \Bigl|\frac{1}{n}\sum_{t\in[n]} Q_{\tilde{X}_t}(a) - Q_{\tilde{X}}(a) \Bigr|\label{TV_estimate_2} \\
&= \sum_{a\in \tilde \calX} \biggl|\frac{1}{n}\sum_{t\in[n]} \Bigl(\sum_{\tilde{x}^n\in \calT_{\delta}^{(n)}(Q_{\tilde{X}|W}|w^n)} \mathrm{U}_{\calT_{\delta}^{(n)}(Q_{\tilde{X}|W}|w^n)}(\tilde{x}^n) Q_{\tilde{X}_t|\tilde{X}^n=\tilde{x}^n}(a)\Bigr)- Q_{\tilde{X}}(a)\sum_{\tilde{x}^n\in \calT_{\delta}^{(n)}(Q_{\tilde{X}|W}|w^n)} \mathrm{U}_{\calT_{\delta}^{(n)}(Q_{\tilde{X}|W}|w^n)}(\tilde{x}^n)\biggl|\label{TV_estimate_3} \\
&= \sum_{a\in \tilde \calX} \biggl|\sum_{\tilde{x}^n\in \calT_{\delta}^{(n)}(Q_{\tilde{X}|W}|w^n)} \mathrm{U}_{\calT_{\delta}^{(n)}(Q_{\tilde{X}|W}|w^n)}(\tilde{x}^n)\hat{T}_{\tilde x^n}(a) - \sum_{\tilde{x}^n\in \calT_{\delta}^{(n)}(Q_{\tilde{X}|W}|w^n)} \mathrm{U}_{\calT_{\delta}^{(n)}(Q_{\tilde{X}|W}|w^n)}(\tilde{x}^n)Q_{\tilde{X}}(a)\biggr|\label{TV_estimate_4} \\
&\leq \sum_{a\in \tilde \calX}\sum_{\tilde{x}^n\in \calT_{\delta}^{(n)}(Q_{\tilde{X}|W}|w^n)} \mathrm{U}_{\calT_{\delta}^{(n)}(Q_{\tilde{X}|W}|w^n)}(\tilde{x}^n)\Bigl|\hat{T}_{\tilde x^n}(a) - Q_{\tilde{X}}(a)\Bigr|\label{TV_estimate_5}\\
&\leq \sum_{a\in \tilde \calX}\sum_{\tilde{x}^n\in \calT_{\delta}^{(n)}(Q_{\tilde{X}|W}|w^n)} \mathrm{U}_{\calT_{\delta}^{(n)}(Q_{\tilde{X}|W}|w^n)}(\tilde{x}^n) \delta Q_{\tilde X}(a)\label{TV_estimate_6}\\
&= \delta, \label{TV_estimate_7}
\end{align} 
where \eqref{TV_estimate_1} follows from \eqref{checkX}, \eqref{TV_estimate_2} is due to the definition of TV distance, \eqref{TV_estimate_3} follows from the law of total probability, \eqref{TV_estimate_4} follows from \eqref{empirical distribution_1} and \eqref{empirical distribution_2}, \eqref{TV_estimate_5} is due to the triangle inequality, \eqref{TV_estimate_6} follows from \eqref{TV_4}, and \eqref{TV_estimate_7} follows from the fact that the probabilities sum to $1$.
Since $d_1(P_X,Q)$ is continuous in $Q$ over the interior of the probability simplex defined on $\tilde{\mathcal{X}}$, it satisfies the Lipschitz condition with some Lipschitz constant $h>0$ when $\delta$ is sufficiently close to zero. For $t\in[n]$, we have:
\begin{align}
\bigl|d_1(P_X,Q_{\hat{X}_t}) -d_1(P_X,Q_{\tilde{X}})\bigr| 
&\leq h\times d_{\mathrm{TV}}( Q_{\hat{X}_t}, Q_{\tilde{X}})\label{differ_perception_1}\\ 
&\leq h \times \delta\label{differ_perception_2}\\
&\leq h \frac{\varepsilon}{2\bigl(H(\tilde X|W)+1\bigr)},\label{differ_perception_3}
\end{align}
where \eqref{differ_perception_2} follows from \eqref{TV_estimate_7}, and  \eqref{differ_perception_3} follows from the choice $\delta \le \frac{\varepsilon}{2(H(\tilde{X}|W)+1)}$. 
Therefore, defining $\kappa_1:=\frac{h}{2\left(H(\tilde X|W)+1\right)}-1$, for $t\in[n]$, we obtain
\begin{align}
d_1(P_X,Q_{\hat{X}_t}) 
&\leq d_1(P_X,Q_{\tilde{X}}) +h \frac{\varepsilon}{2\bigl(H(\tilde X|W)+1\bigr)}\label{perception_est_1} \\
&\leq P_1 + \varepsilon\Bigl(\frac{h}{2\bigl(H(\tilde X|W)+1\bigr)}-1\Bigr)\label{perception_est_2}\\
&= P_1 +\varepsilon\kappa_1, \label{eq:perception_est_final}
\end{align}
where \eqref{perception_est_1} follows from \eqref{differ_perception_3}, and \eqref{perception_est_2} follows from \eqref{assump1_3}. Note that $\kappa_1$ is a constant determined by $H(\tilde{X}|W)$ and the Lipschitz constant $h$. The term $\varepsilon\kappa_1$ vanishes as $\varepsilon \to 0$, implying that the perception constraint holds asymptotically.

Now, we bound the rates $R_{0,\rm{cr}}$, $R_{1,\rm{cr}}$, and $R_{2,\rm{cr}}$. Recalling the definitions $M_0:=\lfloor\exp(n\left(I(X,Y;W) + 2\delta_1\right))\rfloor$, $M_1:=\lfloor \exp({n ( I(X;\tilde X|W) + 2\delta^{\prime})}) \rfloor$, and the parameters $\delta_1=\delta\bigl(H(W) + H(W|XY)\bigr)$ and  $\delta^{\prime}=\delta(H(\tilde{X}|W)+1)$, we obtain 
\begin{align}
R_{0,\rm{cr}}
&= \frac{1}{n}\log{M_0}\label{R0,cr,1}\\
&\leq I(X,Y;W) + 2\delta_1,
\end{align}
where $\delta_1\to 0$ as $\delta\to0$. Furthermore, we have
\begin{align}
R_{1,\rm{cr}} 
&= \frac{1}{n} \log{M_1}\\ 
&\leq I(X;\tilde{X}|W)+2\delta'\\
&\leq  I(X;\tilde{X}|W)+\varepsilon\label{R1,cr,0}\\
&\leq R_{X|W}(Q_{XW}, D_1-2\varepsilon,P_1-2\varepsilon)+  2\varepsilon\label{R1,cr,1}\\
&\le R_{X|W}(Q_{XW},D_1, P_1) +  3\varepsilon,\label{R1,cr,2}
\end{align}
where \eqref{R1,cr,0} follows from the condition $\delta\leq \frac{\varepsilon}{2(H(\tilde{X}|W)+1)}$, \eqref{R1,cr,1} follows from \eqref{assump1_1}, and \eqref{R1,cr,2} holds for sufficiently small $\varepsilon$ due to the continuity of $R_{X|W}(Q_{XW},D_1, P_1)$ with respect to $D_1>0$ and $P_1>0$. Similarly, we obtain $R_{2,\rm{cr}} \leq R_{Y|W}(Q_{YW},D_2,P_2)+3\varepsilon.$\\
\subsubsection{Common Randomness Removal}
In the following analysis, we aim to eliminate the common randomness to obtain a deterministic coding scheme. We perform the analysis for the source sequence $X^n$ since the analysis for $Y^n$ is analogous. The main idea is to simulate $K$ using a negligible fraction of the source sequence. Define $n_0 :=\lfloor n\alpha \rfloor$, where $\alpha>0$ is a sufficiently small constant. The feasibility of choosing such a small $\alpha$ relies on the key observation that $K$ can be simulated using a negligible fraction of source symbols as $H(K)$ is sublinear in $n$. Let $P_{\max} := \max_{(x,y) \in \mathcal{X}\times\mathcal{Y}} P_{XY}(x,y)$, and the maximum probability of $P_{XY}^{n_0}$ is $P_{\max}^{n_0}$, which tends to 0 exponentially as $n \to \infty$. The construction in \cite[Appendix B]{chen_2022_on-the-rate-distortion-perception-function} can be readily extended to the joint random variable $(X,Y)$. Therefore, for $b\in[0:n-1]$, there exists a map $\omega: \mathcal{X}^{n_0}\times\mathcal{Y}^{n_0} \to [0 : n-1]$ satisfying
\begin{align}
\bigl|Q_{\omega(X_{n+1}^{n+n_0},Y_{n+1}^{n+n_0})}(b) - \tfrac{1}{n}\bigr| \le P_{\max}^{n_0}.\label{sim_K_uniform}
\end{align}
Such a mapping $\omega$ guarantees that we can simulate $K$ that is uniformly distributed over $[0 : n-1]$ using $n_0$ symbols.

We then construct the deterministic encoder and decoder as follows:

\noindent\textbf{Encoder:} The encoder operates in two steps. First, it applies the above simulation code on the $(n+1)$-th to $(n+n_0)$-th source symbols to generate $\tilde{K}$. Second, $f_0$, $f_1$ and $f_2$ operate to generate the messages $s_0$, $s_1$ and $s_2$ respectively, i.e.,
\begin{align}
\tilde{K} &= \omega\left(X_{n+1}^{n+n_0},Y_{n+1}^{n+n_0}\right), \\
s_i &= f_i\bigl(\pi_{- \tilde{K}}(X^n,Y^n)\bigr), \,i\in[0:2].
\end{align}
\noindent\textbf{Decoder:} Upon receiving $\tilde{K}$, $s_0$, $s_1$ and $s_2$, the decoders $\phi_1$ and $\phi_2$ generate the reproduced sequences as follows:
\begin{align}
\check{X}^n &= \pi_{\tilde{K}}\bigl(\phi_1(s_0,s_1)\bigr), \\
\check{Y}^n &= \pi_{\tilde{K}}\bigl(\phi_2(s_0,s_2)\bigr), \\
(\check{X}_{n+j},
\check{Y}_{n+j}) &:= (\check{X}_j,\check{Y}_j), \,j \in [n_0].\label{decoder_n_0}
\end{align}
For $t\in[n]$, $a\in\cal\tilde{X}$ and $b\in[0:n-1]$, it follows from the construction that
\begin{align}
Q_{\tilde{K}}(b)&=Q_{\omega(X_{n+1}^{n+n_0},Y_{n+1}^{n+n_0})}(b),\label{simu_K}\\
Q_{\check{X}_t|\tilde K}(a|b) &= Q_{\hat X_t | K}(a|b).\label{check_hat}
\end{align}
Then, we show that the deterministic code satisfies the desired properties:
\begin{align}
\frac{1}{n}\sum_{t\in[n]} \bbE[\Delta_1(X_t, \check{X}_t)] 
&\leq D_1-\frac{\varepsilon}{2},\label{nr_distortion_1}
\end{align}
and
\begin{align}
\bbE[\Delta_1(X_{n+j}, \check{X}_{n+j})] &\leq \bbE_{P_X} \Bigl[ \max_{\tilde{x} \in \tilde{\mathcal{X}}} \Delta_1(X, \tilde{x}) \Bigr] < \infty, \;j \in [n_0],\label{max_dis}
\end{align}
where \eqref{nr_distortion_1} follows from \eqref{distortion_measure_1}--\eqref{distortion_measure_8}. 

Consequently, we obtain:
\begin{align}
\frac{1}{n + n_0} \sum_{t\in[n+n_0]} \bbE[\Delta_1(X_t, \check{X}_t)]
&=\frac{1}{n + n_0} \sum_{t\in[n]}\bbE[\Delta_1(X_t, \check{X}_t)]+\frac{1}{n + n_0} \sum_{j\in[n_0]} \bbE[\Delta_1(X_{n+j}, \check{X}_{n+j})]\\
&\leq \frac{n}{n + n_0}(D_1-\frac{\varepsilon}{2}) + \frac{1}{n + n_0} \sum_{j\in[n_0]} \bbE[\Delta_1(X_{n+j}, \check{X}_{n+j})]\label{reconstruction_D1_1}\\
&\leq \frac{n}{n + n_0}(D_1-\frac{\varepsilon}{2}) + \frac{n_0}{n + n_0}\bbE_{P_X} \Bigl[ \max_{\tilde{x} \in \tilde{\mathcal{X}}} \Delta_1(X, \tilde{x}) \Bigr]\label{reconstruction_D1_2} \\
&=\frac{1}{1 + \alpha}(D_1-\frac{\varepsilon}{2}) + \frac{\alpha}{1 + \alpha}\bbE_{P_X} \Bigr[ \max_{\tilde{x} \in \tilde{\mathcal{X}}} \Delta_1(X, \tilde{x}) \Bigr]\label{reconstruction_D1_3}\\
&\leq D_1 -\frac{\varepsilon}{2} + \frac{\varepsilon}{2}\label{reconstruction_D1_4}\\
&= D_1,
\end{align}
for sufficiently large $n$ and sufficiently small $\alpha>0$. Here \eqref{reconstruction_D1_1} follows from \eqref{nr_distortion_1}, \eqref{reconstruction_D1_2} follows from \eqref{max_dis}, \eqref{reconstruction_D1_3} follows from the definition $n_0 :=\lfloor n\alpha \rfloor$, and \eqref{reconstruction_D1_4} is obtained by choosing $\alpha$ small enough such that $\frac{\alpha}{1 + \alpha}\bbE_{P_X} \left[ \max_{\tilde{x} \in \tilde{\mathcal{X}}} \Delta_1(X, \tilde{x}) \right]\leq \varepsilon/2 $.

Next, for $t\in[n]$ and any $\varepsilon>0$, we can obtain
\begin{align}
d_{\rm{TV}}(Q_{\check X_t}, Q_{\hat{X}_t}) 
&= \sum_{a\in \tilde \calX} \bigl| Q_{\check{X}_t}(a) - Q_{\hat{X}_t}(a) \bigr|\label{TV_sour_rec_1} \\
&=\sum_{a\in \tilde \calX}\Bigl| \sum_{b\in[0:n-1]}\bigl(Q_{\check{X}_t|\tilde{K}}(a|b)Q_{\tilde{K}}(b) - Q_{\hat{X}_t|K}(a|b)Q_{K}(b)\bigr)\Bigr|\label{TV_sour_rec_2}\\
&= \sum_{a\in \tilde \calX}\Bigl| \sum_{b\in[0:n-1]}Q_{\check{X}_t|\tilde{K}}(a|b)\bigl(Q_{\tilde{K}}(b) - Q_{K}(b)\bigr) \Bigr|\label{TV_sour_rec_3}\\
&= \sum_{a\in \tilde \calX}\Bigl| \sum_{b\in[0:n-1]}Q_{\check{X}_t|\tilde{K}}(a|b)\left(Q_{\omega(X_{n+1}^{n+n_0},Y_{n+1}^{n+n_0})}(b) - \tfrac{1}{n}\right) \Bigr|\label{TV_sour_rec_4}\\
&\leq \sum_{a\in \tilde \calX,b\in[0:n-1]}Q_{\check{X}_t|\tilde{K}}(a|b)\Bigl|Q_{\omega(X_{n+1}^{n+n_0},Y_{n+1}^{n+n_0})}(b) - \tfrac{1}{n} \Bigr|\label{TV_sour_rec_5}\\
&\leq \sum_{a\in \tilde \calX,b\in[0:n-1]}Q_{\check{X}_t|\tilde{K}}(a|b)P_{\max}^{n_0}\label{TV_sour_rec_6}\\
&=nP_{\max}^{n_0}\label{TV_sour_rec_7}\\
&\leq \frac{\varepsilon}{2},\label{TV_sour_rec_8}
\end{align}
where \eqref{TV_sour_rec_1} is due to the definition of TV distance, \eqref{TV_sour_rec_2} follows from the law of total probability, \eqref{TV_sour_rec_3} holds due to \eqref{check_hat}, \eqref{TV_sour_rec_4} follows from \eqref{simu_K} and the fact that $K$ is uniformly distributed over $[0:n-1]$, \eqref{TV_sour_rec_5} follows from the triangle inequality, \eqref{TV_sour_rec_6} is due to \eqref{sim_K_uniform}, \eqref{TV_sour_rec_7} follows from the fact that probability mass functions sum to $1$, and \eqref{TV_sour_rec_8} holds because $P_{\max}^{n_0}$ tends to 0 exponentially as $n \to \infty$.\\
Since $d_1(P_X,Q)$ is continuous in $Q$ over the interior of the probability simplex, it satisfies the Lipschitz condition with some Lipschitz constant $h_2>0$ when $\delta$ is sufficiently close to zero. We have:
\begin{align}
|d_1(P_X, Q_{\check{X}_t}) -  d_1(P_X, Q_{\hat{X}_t})|\leq h_2 \times d_{\rm{TV}}(Q_{\check X_t}, Q_{\hat{X}_t}), \; t \in [n + n_0].\label{n+n_0_distribution}
\end{align}
Note that the validity of \eqref{n+n_0_distribution} for the extended indices $t \in [n+1 : n+n_0]$ follows directly from the construction of the deterministic decoder in \eqref{decoder_n_0}. Specifically, by setting $\check{X}_{n+j} := \check{X}_j$ for $j \in [n_0]$, we ensure that the marginal distribution at time $n+j$ is identical to that at time $j$, i.e., $Q_{\check{X}_{n+j}} = Q_{\check{X}_j}$. The inequality naturally extends to the entire block length $n+n_0$. Now, defining $\kappa_3:=\frac{h}{H(\tilde X|W)+1}+h_2-2$, we obtain for $t\in[n+n_0]$:
\begin{align}
d_1(P_X, Q_{\check{X}_t}) &\leq d_1(P_X, Q_{\hat{X}_t}) +  h_2 \times d_{\rm{TV}}(Q_{\check X_t}, Q_{\hat{X}_t})\label{recon_perception_0}\\
&\leq P_1 +\varepsilon\Bigl(\frac{h}{2\bigl(H(\tilde X|W)+1\bigr)}-1\Bigr) +  h_2 \times d_{\rm{TV}}(Q_{\check X_t}, Q_{\hat{X}_t})\label{recon_perception_1}\\
&\leq P_1 + \varepsilon \Bigl(\frac{h}{2\bigl(H(\tilde X|W)+1\bigr)}-1\Bigr) +h_2\frac{\varepsilon}{2}\label{recon_perception_2}\\
&= P_1 +\frac{\varepsilon}{2}\Bigl(\frac{h}{H(\tilde X|W)+1}+h_2-2\Bigr)\\
&= P_1 + \frac{\varepsilon}{2}\kappa_3,
\end{align}
where \eqref{recon_perception_0} follows from \eqref{n+n_0_distribution}, \eqref{recon_perception_1} follows from \eqref{perception_est_2}, and \eqref{recon_perception_2} follows from \eqref{TV_sour_rec_8}. The term $\frac{\varepsilon}{2}\kappa_3$ vanishes as $\varepsilon \to 0$, implying that the perception constraint holds asymptotically.

Following the derivation steps in \eqref{R0,cr,1}--\eqref{R1,cr,2}, we obtain
\begin{align}
R_{0,\rm{nr}} &= \frac{1}{n+n_0}(\log{M_0}+\log n)\\
&\leq I(X,Y;W) + 2\delta_1+\frac{\log{n}}{n+n_0} \\
&= I(X,Y;W) + 2\delta(H(W)+H(W|XY))+\frac{\log{n}}{n+n_0},\label{R_0_nr_1}
\end{align}
and
\begin{align}
R_{1,\rm{nr}} 
&= \frac{1}{n+n_0}\left( \log{M_1} + \log{n}\right)\\ 
&\leq \frac{n}{n+n_0}\left( R_{X|W}(Q_{XW}, D_1-2\varepsilon,P_1-2\varepsilon) + 2\varepsilon\right) + \frac{\log{n}}{n+n_0}\label{12}\\
&\le R_{X|W}(Q_{XW}, D_1-2\varepsilon,P_1-2\varepsilon) + 2\varepsilon + \frac{\log{n}}{n+n_0}\\
&\le R_{X|W}(Q_{XW},D_1, P_1) +  3\varepsilon + \frac{\log{n}}{n+n_0},\label{13}
\end{align}
where \eqref{R_0_nr_1} follows from the definition $\delta_1=\delta(H(W)+H(W|XY))$. Note that the term $\frac{\log{n}}{n+n_0}\to0$ as $n\to\infty$. 
Similarly, we obtain $R_{2,\rm{nr}} \leq R_{Y|W}(Q_{YW},D_2, P_2) +  3\varepsilon + \frac{\log{n}}{n+n_0}$.
\subsection{Converse Proof}
Finally, we establish the converse part. Let $(R_{0,\rm{cr}},R_{1,\rm{cr}},R_{2,\rm{cr}})$ be an achievable rate triple with respect to the distortion and perception constraints $(D_1,D_2,P_1,P_2)$. Let $\hat{X}^n$ be the
reproduced sequence corresponding to $X^n$, and let $\hat{Y}^n$ be the reproduced sequence corresponding to $Y^n$. Assume that the sequences satisfy $\frac{1}{n} \sum_{i=1}^n\bbE_{Q_{X_i\hat X_i}}[\Delta_1(X_i,\hat X_i)]\leq D_1$, $\frac{1}{n} \sum_{i=1}^n\bbE_{Q_{Y_i\hat Y_i}}[\Delta_2(Y_i,\hat Y_i)]\leq D_2$, $d_1(P_X,Q_{\hat{X}_i})\leq P_1$ and $d_2(P_Y,Q_{\hat{Y}_i})\leq P_2$, $i\in[n]$. Then we obtain
\begin{align}
R_{0,\rm{cr}} 
&\geq \frac{1}{n}H(S_0|K)\\
&\geq \frac{1}{n}\left( H(S_0|K) - H(S_0|X^n,Y^n,K)\right)\\
&= \frac{1}{n}I(X^n,Y^n;S_0|K)\\
&=\frac{1}{n}I(X^n,Y^n;S_0|K) + \frac{1}{n}I(X^n,Y^n;K)\label{indep_K}\\
&=\frac{1}{n}I(X^n,Y^n;S_0,K)\\
&\geq \frac{1}{n}I(X^n,Y^n;S_0)\\
&= \frac{1}{n}\sum_{i=1}^n I(X_i,Y_i;S_0|X^{i-1},Y^{i-1})\\
&=\frac{1}{n}\sum_{i=1}^n I(X_i,Y_i;S_0|X^{i-1},Y^{i-1}) + \frac{1}{n}\sum_{i=1}^n I(X_i,Y_i;X^{i-1},Y^{i-1})\label{memoryless_source}\\
&=\frac{1}{n}\sum_{i=1}^n I(X_i,Y_i;S_0,X^{i-1},Y^{i-1}) \\
&=\frac{1}{n}\sum_{i=1}^n I(X_i,Y_i;W_i)\label{eq:Wi_1}\\
&= \frac{1}{n}\sum_{i=1}^n I(X_i,Y_i;W_i|Q=i)\label{eq:Q_1}\\
&= I(X_Q,Y_Q;W_Q|Q)\label{eq:Q_2}\\
&= I(X_Q,Y_Q;W_Q,Q)\label{eq:Q_3}\\
&= I(X,Y;W),\label{eq:Q_4}
\end{align}
where \eqref{indep_K} follows from the fact that $K$ is independent of $(X^n,Y^n)$, \eqref{memoryless_source} is due to the memoryless property of the sources, \eqref{eq:Wi_1} holds because $W_i:=(S_0,X^{i-1},Y^{i-1})$, \eqref{eq:Q_1} holds because $Q$ is uniformly distributed over $[n]$ and independent of $(X^n,Y^n)$, \eqref{eq:Q_4} follows from $W:=(W_Q,Q)$.

Furthermore,
\begin{align}
R_{1,\rm{cr}}
&\geq \frac{1}{n}H\left(S_1|K\right)\\
&\geq \frac{1}{n}H\left(S_1|K,S_0\right)\\
&\geq \frac{1}{n}I(X^n,Y^n;S_1|K,S_0)\\
&\geq \frac{1}{n}I(X^n,Y^n;\hat X^n|K,S_0)\label{eq:DPI_hatX}\\
&=\frac{1}{n}\left(I(X^n,Y^n;\hat X^n|K,S_0) + I(X^n,Y^n;K|S_0) \right)\label{eq:iid_K}\\
&=\frac{1}{n}I(X^n,Y^n;\hat X^n, K|S_0)\\
&\geq \frac{1}{n}I(X^n,Y^n;\hat X^n|S_0)\\
&=\frac{1}{n}\sum_{i=1}^n I(X_i,Y_i;\hat X^n|S_0,X^{i-1},Y^{i-1})\\
&=\frac{1}{n}\sum_{i=1}^n I(X_i,Y_i;\hat X^n|W_i)\label{eq:Wi_2}\\
&\geq \frac{1}{n}\sum_{i=1}^n I(X_i;\hat X_i|W_i)\\
&\geq \frac{1}{n} \sum_{i=1}^nR_{X|W}\left(Q_{XW},\bbE_{Q_{X_i\hat X_i}}[\Delta_1(X_i,\hat X_i)], d_1(P_{X_i},Q_{\hat X_i})\right)\label{converse_RDP}\\
&\geq R_{X|W}\Bigl(Q_{XW},\frac{1}{n} \sum_{i=1}^n\bbE_{Q_{X_i\hat X_i}}[\Delta_1(X_i,\hat X_i)], \frac{1}{n} \sum_{i=1}^nd_1(P_{X_i},Q_{\hat X_i})\Bigr)\label{convexity_RDP_1}\\
&\geq R_{X|W}(Q_{XW},D_1,P_1),
\end{align}
where \eqref{eq:DPI_hatX} follows from the data processing inequality, \eqref{eq:iid_K} follows from the fact that $K$ is independent of $(X^n,Y^n)$, \eqref{converse_RDP} follows from \eqref{conditional RDP}, \eqref{eq:Wi_2} holds because $W_i:=(S_0,X^{i-1},Y^{i-1})$, and \eqref{convexity_RDP_1} follows from the convexity of $R_{X|W}(Q_{XW},D_1,P_1)$ with respect to $(D_1,P_1)$~\cite[Proposition 2]{Salehkalaibar_2024_Rate-Distortion-Perception-Tradeoff-Based-on-the-Conditional-Distribution-Perception-Measure} and Jensen’s inequality~\cite[Theorem 2.6.2]{Thomas_2006_Elements-of-information-theory}.
Similarly, we have $R_{2,\rm{cr}}\geq R_{Y|W}(Q_{YW},D_2,P_2)$. 

It is straightforward to establish the following ordering by using the operational interpretations of the relevant quantities~\cite[Remark 2]{chen_2022_on-the-rate-distortion-perception-function}:
\begin{align}
R_{i,\rm{cr}} \leq R_{i,\rm{pr}} \leq R_{i,\rm{nr}},\,i \in [0:2].
\end{align}
Thus, we have:
\begin{align}
R_{0,\rm{nr}} \geq R_{0,\rm{cr}} \geq I(X,Y;W),
\end{align}
\begin{align}
R_{1,\rm{nr}} \geq R_{1,\rm{cr}} \geq R_{X|W}(Q_{XW},D_1, P_1),
\end{align}
and
\begin{align}
R_{2,\rm{nr}} \geq R_{2,\rm{cr}} \geq R_{Y|W}(Q_{YW},D_2, P_2).
\end{align}
This completes the proof.
\end{document}